\documentclass[11pt]{article}
\usepackage{amsmath}
\usepackage{amsthm}
\usepackage{fancyhdr}
\usepackage{amsfonts}
\usepackage{amssymb}
\usepackage{enumerate}
\usepackage{tikz}
\usepackage{subfigure}
\usepackage{xcolor}
\usepackage{graphicx}
\usepackage{tabularx}
\usepackage[center]{caption}
\usepackage{array}
\usepackage{multirow}
\usepackage{float}
        \theoremstyle{plain}
        
        \newtheorem{proposition}{Proposition}[section]

        \theoremstyle{remark}
        \newtheorem{remark}{\bf Remark}[section]
        \theoremstyle{remark}
        
        \theoremstyle{remark}

\newcommand{\la}{\left\langle}
\newcommand{\ra}{\right\rangle}
\newcommand{\mo}{tr}
\DeclareMathOperator{\Tr}{Tr}
\newcommand{\Kn}{{\rm Kn}}
\newcommand{\R}{{\mathbb{R}}}
\newcommand{\cint}[1]{\langle #1 \rangle}
\newcommand{\ent}{\mathbb{H}}
\newcommand{\entred}{{\cal H}}
\newcommand{\F}{{\bf F}}

\graphicspath{{./Figures/}}

\begin{document}

\begin{center}
{\bf An ES-BGK model for vibrational polyatomic  gases}

\vspace{1cm}
Y. Dauvois$^{1}$, J. Mathiaud$^{1,3}$,  L. Mieussens$^2$ \\

\bigskip
$^1$CEA-CESTA\\
15 avenue des sabli\`eres - CS 60001\\
33116 Le Barp Cedex, France\\

\bigskip
$^2$Univ. Bordeaux, Bordeaux INP, CNRS, IMB, UMR 5251, F-33400 Talence, France.\\
{ \tt(Luc.Mieussens@math.u-bordeaux.fr)}

\bigskip
$^3$Univ. Bordeaux, CNRS, CELIA, UMR 5107, F-33400 Talence, France.\\
{ \tt(julien.mathiaud@u-bordeaux.fr))}

\end{center}

\abstract{We propose an extension of the Ellipsoidal-Statistical BGK
  model to account for discrete levels of vibrational energy in a
  rarefied polyatomic gas. This model satisfies an H-theorem and
  contains parameters that allow to fit almost arbitrary values for
  the Prandtl number and the relaxation times of rotational and
  vibrational energies. With the reduced distribution technique, this
  model can be reduced to a three distribution system that could be
  used to simulate polyatomic gases with rotational and vibrational
  energy for a computational cost close to that of a simple monoatomic gas.}

\tableofcontents

\section{Introduction}
During reentry a space vehicle encounters several atmospheric layers
at high velocity; it is critical to estimate heat fluxes to design its
heat shield. At high altitudes, air is in rarefied regime and usual
macroscopic fluid dynamics equations become non valid; instead the
Boltzmann equation is used to describe transport and collisions of
molecules at a microscopic scale. The Direct Simulation Monte Carlo
method (DSMC)~\cite{bird,BS_2017} is generally used but its
computational cost is known to become very large close to dense
regimes. In this case, it can be more efficient to use deterministic
solvers based on discretizations of BGK like models of the Boltzmann
equation: the Boltzmann collision operator is replaced by a simple
relaxation operator towards Maxwellian equilibrium which satisfies
conservation of macroscopic quantities and second principle of
thermodynamics. However, by construction, the simple BGK
model~\cite{bgk} (derived for monoatomic gases) induces a Prandtl
number equal to $1$ and cannot predict the correct transport
coefficients: models which include another parameter to uncouple the
thermal relaxation from the viscosity relaxation have been proposed,
like the ES-BGK model~\cite{holway} and the Shakhov
model~\cite{S_model}.  Both models have been extended to polyatomic
gases with degrees of freedom of rotation~\cite{R_model,esbgk_poly}. However, up to our
knowledge, only the ES-BGK model can be proved to satisfy the second
principle of thermodynamics (also called H-theorem in kinetic
theory). This was proved and extended to polyatomic gases with
rotational energy by Andries et al.~\cite{esbgk_poly}. We also mention
another model where the Boltzmann collision operator is replaced by a
Fokker-Planck operator in velocity variable that allows for efficient
stochastic simulations: this model has been recently extended by Jenny
et al.~\cite{Jnny2010,Grj2011,Grj2012,GJ_13} and also by Mathiaud and
Mieussens~\cite{Mathiaud2016,Mathiaud2017,Mathiaud2019}.

Here we want to extend the ES-BGK model of~\cite{esbgk_poly} to take
into account vibration energy of molecules. Indeed, at high
temperature, there are exchanges of energy between translational,
rotational, and vibration modes. Taking into account vibration energy
has a strong influence on the parietal heat flux and shock position
\cite{mathiaudhdr,Mathiaud2018}. In recent literature, one can find
models that take into account vibrations of molecules by assuming a
continuous distribution of the vibrational energy
\cite{RS_2014,WYLX_2017,ARS_2017,KKA_2019,Mathiaud2018}. However, up
to our knowledge, it is not possible so far to prove any H-theorem for
these models.

Moreover, while
transitional and rotational energies in air can be considered as
continuous for temperature larger than 1K and 10K, respectively,
vibrational energy can be considered as continuous only for much
larger temperatures (2000K for oxygen and 3300K for nitrogen). For
flows up to 3000K around reentry vehicles, discrete levels of
vibrational energy must be used~\cite{anderson}.
An older BGK model proposed by Morse
\cite{Morse_1964} accounts for vibrational effects through discrete
energy levels of vibration. We used this idea to derive a new BGK
model with discrete vibrational energy levels
in~\cite{Mathiaud2019}, for which we were able to prove a
H-theorem. In this paper, we use this model and the methodology
of~\cite{esbgk_poly} to propose an ES-BGK extension, for which we are
also able to prove a H-theorem. This model contains some free parameters
that can be adjusted to recover any relaxation times for rotation and
vibration modes (as given by Jeans and Landau-Teller equations, for
instance), as well as the correct value of the Prandtl number. Note that since the vibration energy
is a non linear function of temperature, this extension is not
trivial: while~\cite{esbgk_poly} is based on convex combinations of
temperatures, we have found more natural to work with convex combinations of
energies.

At a computational level, note that even if the computational cost of a deterministic solver based
on a model with so many variables (velocity, energy of rotation and
vibration) is necessarily very large, the great advantage of the BGK
approach is that this cost can be drastically reduced. Indeed, like
every BGK models, the computational complexity of our new model can be
reduced by the standard reduced distribution
technique~\cite{HH_1968}: this gives a model that has the same
computational cost as a model for monoatomic gas (the only kinetic
variable is the velocity), while it still accounts for rotation and
vibration energy exchanges. Moreover, a H-theorem also holds for this reduced model.

The outline of our paper is as follows.  The next two sections are
necessary to prepare the introduction of our model: in section
\ref{sec:thermo} we detail the different energies at macroscopic scale
as functions of temperature, and we give their mathematical
properties; the description, at the kinetic level, of a polyatomic gas
with energy of translation, rotation, and vibration is given in
section~\ref{sec:dist_funct}. We define our new ES-BGK model in
section \ref{sec: esbgk}, in which we also prove a H-theorem. In
section \ref{sec: energ}, we show how the parameters of our model can
be adjusted to fit the correct relaxation times of rotation and
vibration. In section \ref{sec: chap} we derive the hydrodynamics
limits of our model by the usual Chapman-Enskog expansion. The reduced
ES-BGK model is derived and analyzed in section~\ref{sec:
  r_esbgk}. Finally, some preliminary numerical results are shown in
section~\ref{sec: numeric} to illustrate the capability of our model
to capture correct relaxation times.

\section{\label{sec:thermo}Internal energies of vibrational polyatomic perfect gases}

\subsection{The different macroscopic internal energies at equilibrium}

In these paper we consider vibrational polyatomic perfect gases. Each
molecule has several degrees of freedom: translation, rotation and
vibration. At the macroscopic level, a gas in thermodynamical equilibrium
at temperature $T$ has different specific energies associated to each
mode. For translational and rotational modes, the translational and
rotational energies are
\begin{equation} \label{eq-etr_erot} 
  e_{tr}(T)=\frac32 RT, \quad \text{ and } \quad 
e_{rot}(T)=\frac{\delta}2 RT,
\end{equation}
where $\delta$ is the number of degrees of freedom of rotation.
For the vibrational mode, the vibrational energy is
\begin{equation}\label{eq-evib} 
  e_{vib}(T)=\frac{RT_0}{\exp\left(T_0/T\right)-1},
\end{equation}
where $T_0$ is some characteristic temperature of the vibrations
($T_0=2256 K$ for dioxygen for instance).

The total internal energy is denoted by $e(T)$ and is simply the sum of
the three previous energies: 
\begin{equation}  \label{eq-e}
e(T) = e_{tr}(T)+e_{rot}(T)+e_{vib}(T).
\end{equation}

Finally, we also define the joint translational-rotational energy
function
\begin{equation}  \label{eq-etrrot}
e_{tr,rot}(T) =  e_{tr}(T) + e_{rot}(T) = \frac{3 +\delta }{2} RT,
\end{equation}
 that will be useful for the derivation of our model.

\subsection{Mathematical properties of the energy functions}

For the construction of our model, it is useful to study the
specific internal energies defined in the previous section, as
functions of the temperature. The property needed here is the
invertibility, since it will be used to define an equivalent
temperature for each mode in non-equilibrium regimes.

We denote by $e_i^{-1}$ the function that maps any given energy $E$ to
the corresponding temperature. In other words, 
$ e_i^{-1}(E) = T$ such that $e_i(T) = E$, 
where $i$ stands for $tr$, $rot$, $vib$, and $tr,rot$. Since $e_{tr}$,
$e_{rot}$, and $e_{tr,rot}$
are linear functions of $T$ (see~\eqref{eq-etr_erot} and~\eqref{eq-etrrot}), they are clearly
invertible, and we have
\begin{equation}\label{eq-inv_etr_rot} 
  e_{tr}^{-1}(E) = \frac{2}{3R}E, \qquad  
e_{rot}^{-1}(E) =  \frac{2}{\delta R}E, \quad
\text{ and } \quad  e_{tr,rot}^{-1}(E) =  \frac{2}{(3+\delta) R}E.
\end{equation}
For $e_{vib}$, which is a non linear function of $T$, it can be proved
it is increasing, thus invertible, and we
have
\begin{equation}\label{eq-inv_evib} 
  e_{vib}^{-1}(E) = {T_0}/ {\log\left( 1+\frac{RT_0}{E} \right)}.
\end{equation}
The total internal energy $e$ is also an increasing function (see~\eqref{eq-e}), thus
invertible, but its inverse $e^{-1}(E)$ cannot be written
analytically, and instead it must be computed
numerically. In other words
\begin{equation}\label{eq-inv_e} 
   e^{-1}(E) = T \quad \text{ such that } \quad E = \frac{3+\delta}{2}RT  + \frac{RT_0}{\exp\left(T_0/T\right)-1}
\end{equation}
which has to be solved numerically.

\section{Kinetic description}
\label{sec:dist_funct} 

\subsection{Distribution function}

The state of any gas molecule will be described by its position $x$,
its velocity $v$, its rotational energy $\varepsilon$, and its
discrete vibrational energy. In the case of the usual simple
harmonic oscillator model, this energy is given by $iRT_0$, where $i$ is the $i$th vibrational energy
level and $T_0$ is the characteristic vibrational temperature of the
gas.

The distribution function of the gas is the mass density
$f(t,x,v,\varepsilon,i)$ of molecules that at time
$t$ are located in a elementary volume $dx$ centered in $x$, have the
velocity $v$ in a elementary volume $dv$, have the rotational energy
$\varepsilon$ centered in $d\varepsilon$ and the discrete vibrational
energy $iRT_0$.

The macroscopic densities of mass $\rho$, momentum $\rho u$, and
internal energy $\rho E$ are defined
by the first five moments of $f$:
\begin{equation}\label{eq-mtsf} 
\rho=\la f \ra_{v,\varepsilon,i},\qquad \rho u =\la v f \ra_{v,\varepsilon,i},\qquad \rho E(f)=\la \left(\frac{1}{2}|v-u|^2+\varepsilon+iRT_0\right)f \ra_{v,\varepsilon,i}.
\end{equation}
In this paper, to clarify the notations, the dependence of $E$ on $f$
is made explicit, and we denote by $\la \phi
\ra_{v,\varepsilon,i}(t,x)=\sum_{i=0}^{+\infty}\int_{\mathbb{R}^3}\int_{\mathbb{R}}\phi(t,x,v,\varepsilon,i)d\varepsilon
dv$ the integral of any function $\phi$.

The specific internal energy $E(f)$ can be decomposed
into
\begin{equation}  \label{eq-EEtrErotEvib}
E(f) = E_{tr}(f) +  E_{rot}(f) + E_{vib}(f),
\end{equation}
which is the sum of the energy $E_{tr}(f)$ associated with the translational motion of
particles, the energy $E_{rot}(f)$ associated with the rotational mode,
and the energy $E_{vib}(f)$ associated with the vibrational mode, defined by
\begin{equation}\label{eq-Etrrotvib} 
\rho E_{tr}(f)=\la \frac{1}{2}|v-u|^2 f \ra_{v,\varepsilon,i},\qquad 
\rho E_{rot}(f)=\la \varepsilon f \ra_{v,\varepsilon,i},\qquad
\rho E_{vib}(f)=\la iRT_0 f \ra_{v,\varepsilon,i}.
\end{equation}

We also define the shear stress tensor $\Theta$ and the heat flux $q$ by
\begin{equation}\label{eq-Theta_q} 
\rho \Theta=\la (v-u)\otimes (v-u) f \ra_{v,\varepsilon,i}
\qquad 
q=\la \left(\frac{1}{2}|v-u|^2 + \varepsilon + iRT_0\right) (v-u) f \ra_{v,\varepsilon,i}.
\end{equation}

\subsection{Internal temperatures}
\label{subsec:int_temp}

When the gas is in a non-equilibrium state, as described by the
distribution $f$, a temperature can be defined for each mode, by using
the specific energy functions and their inverse as defined in
section~\ref{sec:thermo}. Indeed, the translational, rotational, and
vibrational temperatures are defined by
\begin{equation}\label{eq-T_tr_rot_vib} 
  T_{tr} = e_{tr}^{-1}(E_{tr}(f)), \quad 
  T_{rot} = e_{rot}^{-1}(E_{rot}(f)), \quad 
  T_{vib} = e_{vib}^{-1}(E_{vib}(f)), 
\end{equation}
so that we have the following relations
\begin{equation}\label{eq-E_T} 
E_{tr}(f)=\frac{3}{2}RT_{tr},\quad 
E_{rot}(f)=\frac{\delta}{2}RT_{rot},\quad 
E_{vib}(f)=\frac{RT_0}{\exp(T_0/T_{vib})-1}.
\end{equation}

The equilibrium temperature $T_{eq}$  is the temperature corresponding
to the total internal energy, that is to say 
\begin{equation}\label{eq-Teq} 
  T_{eq}  = e^{-1}(E(f)).
\end{equation}
In other words, $T_{eq}$ can be obtained by numerically solving
\begin{equation}\label{eq-ETeq} 
  E(f) = \frac{3+\delta}{2}RT_{eq}  + \frac{RT_0}{\exp\left(T_0/T_{eq}\right)-1}.
\end{equation}

Note that~\eqref{eq-Theta_q} and~\eqref{eq-E_T} give the following relation between $T_{tr}$ and the
stress tensor:
${\rm Tr}(\Theta) = 3RT_{tr}$. Moreover, each diagonal component of $\Theta$
can be associated to a directional translational temperature: indeed,
the translational temperature $T_{j,j}$ in direction $j$ can be defined
by $\Theta_{jj} = RT_{j,j}$, where $j=1,2,3$. Consequently, the
previous relation gives
$T_{tr} =  (T_{1,1}+T_{2,2}+T_{3,3})/3$.

Finally, it is useful for the following to define the intermediate
translational-rotational temperature by
\begin{equation}  \label{eq-T_trrot}
T_{tr,rot} = e_{tr,rot}^{-1}(E_{tr}(f)+E_{rot}(f)).
\end{equation}

   \subsection{Vibrational number of degrees of freedom}

By analogy with the relation between $E_{rot}(f)$ and $T_{rot}$ (see~\eqref{eq-E_T}), a number
of degrees of freedom $\delta_v(T_{vib})$ for the vibration mode can
be defined such that $E_{vib}(f) =
\frac{\delta_v(T_{vib})}{2}RT_{vib}$, so that we have
\begin{equation}\label{eq-deltav} 
 \delta_v(T_{vib})=\frac{2T_0/T_{vib}}{\exp(T_0/T_{vib})-1}.
\end{equation}
This number is not an integer, is temperature dependent, and tends to $2$ for large
$T_{vib}$.

\subsection{Relaxation times}
\label{sec:relax}

The exchanges of energy between the different modes and their
relaxation to equilibrium are characterized by the relaxation times
$\tau$, $\tau_{rot}$, $\tau_{vib}$.  The first one is the translation
relaxation time, that can be written as $\tau = 1/\nu$, where $\nu$ is
the collision frequency of molecules. The two others are the
rotational and vibrational relaxation times. They can be written as
functions of $\tau$ by $\tau_{rot}=\tau Z_{rot}$ and
$\tau_{vib}=\tau Z_{vib}$, where $Z_{rot}$ and $Z_{vib}$ can be viewed
as average numbers of collisions needed to enforce a change in
rotational and vibrational energy.

In most cases $1< Z_{rot}< Z_{vib}$ and relaxation processes occur in
a specific sequence (see~\cite{bird} for empirical laws that are
temperature dependent): first, the translational
temperatures $T_{j,j}$ in the three directions $j=1,2,3$ relax
towards the mean translational temperature $T_{tr}$, then the
translational and rotational temperatures $T_{tr}$ and $T_{rot}$ relax
towards the intermediate temperature $T_{tr,rot}$, and for longer
times this temperature and the
vibrational temperature $T_{vib}$ relax towards the equilibrium
temperature $T_{eq}$ (see figure~\ref{fig: temp} in section~\ref{sec: numeric} for an illustration).

\section{\label{sec: esbgk}An ES-BGK model with vibrations}

In this section, our new ES-BGK model that accounts for
vibrations of molecules is presented, and its main properties are
stated and discussed.

\subsection{Construction of the model}
\label{subsec:construct}

The evolution of the mass density of a gas in non-equilibrium is
described by the Boltzmann equation
\begin{equation}
\label{eq: cinetique}
\partial_t f+v \cdot\nabla f=Q(f),
\end{equation}
where $Q(f)$ is the Boltzmann collision operator. 

A simpler relaxation BGK like model can be derived, as proposed
in~\cite{Mathiaud2019}, where $Q(f)$ is replaced by
$\frac{1}{\tau}(\mathcal{M}[f] - f) $, where $\tau$ is a relaxation
time and $\mathcal{M}[f]$ is the generalized Maxwellian in velocity
and energy, as defined by
\begin{equation}\label{eq-Mf} 
\mathcal{M}[f](v,\varepsilon,i)= \mathcal{M}_{\mo}[f](v)\mathcal{M}_{rot}[f](\varepsilon)\mathcal{M}_{vib}[f](i),
\end{equation}
where
\begin{align*}
& \mathcal{M}_{\mo}[f](v)=\frac{\rho}{(2\pi RT_{eq})^{3/2}}\exp\left(-\frac{|v-u|^2}{2RT_{eq}}  \right),\\
& \mathcal{M}_{rot}[f](\varepsilon)=\frac{\Lambda(\delta)\varepsilon^{\frac{\delta-2}{2}}}{(R T_{eq})^{\delta/2}}\exp\left( -\frac{\varepsilon}{R T_{eq}} \right),
\\
& {M}_{vib}[f](i)=\left(1-\exp(-T_0/T_{eq})\right)\exp\left(-i\frac{T_0}{T_{eq}} \right),
\end{align*}
where exponential
laws associated to vibrations and rotations are normalized by functions 
$(1-\exp(-T_0/T_{vib}^{rel}))$ and $\Lambda(\delta)=1/{\bf
  \Gamma}(\frac{\delta}{2})$, where ${\bf \Gamma}$ is the usual gamma function.

However, this model is too simple, since the single relaxation time
cannot account for the various time scales of the original
problem. Indeed, such a model gives the same value for rotational and
vibrational relaxation times, and the same value for relaxation times
of viscous and thermal fluxes, which gives the usual incorrect value $\Pr =
1$ of the Prandtl number.

This problem can be fixed by using additional parameters in the model
(at least 3 in this case). The correct Prandtl number for a monoatomic
gas can be obtained by the ES-BGK approach~\cite{holway}, which has
been extended later in~\cite{esbgk_poly} to account for a correct
rotational time scale for polyatomic gases. Here, we extend this model
to account for a correct vibrational time scale. Note that in this
case, since the relation between temperature and energy is non linear,
we find it more relevant to make an intensive use of the energy
variable, that makes the derivation a bit different from that  of~\cite{esbgk_poly}.

Our ES-BGK collision operator is the following:
\begin{equation}
\label{eq: operator}
Q(f)=\frac{1}{\tau}(\mathcal{G}[f]-f),
\end{equation}
where the Gaussian distribution $\mathcal{G}[f]$ is defined by 
\begin{equation}\label{eq: Gauss}
\mathcal{G}[f](v,\varepsilon,i)= \mathcal{G}_{\mo}[f](v)\mathcal{G}_{rot}[f](\varepsilon)\mathcal{G}_{vib}[f](i),
\end{equation}
with
\begin{equation}\label{eq-Gaussalpha} 
\begin{aligned}
&\mathcal{G}_{\mo}[f](v)=\frac{\rho}{  \sqrt{\hbox{det}(2\pi \Pi)}}  \exp\left(-\frac{1}{2} (v-u)^{T}\,\Pi^{-1}\,(v-u))\right),\\
& \mathcal{G}_{rot}[f](\varepsilon)=\frac{\Lambda(\delta)}{(R {T_{rot}^{rel}})^{\delta/2}}\varepsilon^{\frac{\delta-2}{2}}\exp\left( -\frac{\varepsilon}{R T_{rot}^{rel}}\right) , \\
& \mathcal{G}_{vib}[f](i)=(1-\exp(-T_0/T_{vib}^{rel})) \exp\left(-i\frac{T_0}{T_{vib}^{rel}} \right).
\end{aligned}
\end{equation}
Note that $\mathcal{G}_{\mo}[f]$, $\mathcal{G}_{rot}[f]$ and
$\mathcal{G}_{vib}[f]$ are distributions associated to the energies of
translation, rotation and vibration of the molecules. 

The covariance matrix $\Pi$ and the temperatures
$T_{rot}^{rel}$ and $T_{vib}^{rel}$ are modifications of the stress
tensor $\Theta$ and rotational and vibrational temperatures so as to
fit different relaxation times to some given values, as it is
explained below.

First, the corrected stress tensor $\Pi$ is defined by (with $I$ the identity matrix):
\begin{equation}
\label{eq: Pi}
\Pi=
\Gamma R T_{eq} I 
+ (1-\Gamma) \left[ \theta R T_{tr,rot} I 
                    + (1-\theta)(  \nu \Theta + (1-\nu) R T_{tr} I ) \right] ,
\end{equation}
so that the hierarchy of relaxation processes explained in
section~\ref{sec:relax} holds: (1) the directional temperatures
$T_{j,j}$ (the diagonal elements of $\Theta$) first relax to $T_{tr}$
(this is governed by parameter $\nu$); (2) the translational temperature $T_{tr}$ relaxes to the
intermediate temperature $T_{tr,rot}$ (this is governed by
parameter $\theta$); (3)  this
temperature relaxes to the final equilibrium temperature $T_{eq}$, as
governed by parameter $\Gamma$.

Now the relaxation temperatures
$T_{rot}^{rel}$ and $T_{vib}^{rel}$, used in distributions $\mathcal{G}_{rot}$ and
$\mathcal{G}_{vib}$, are defined with the same idea as
the covariance matrix $\Pi$, except that we first write the
relaxations in term of energies. Indeed, we define the relaxation
energies for rotation and vibration by
\begin{equation}
\label{eq: fermeture}
\begin{split}
& e_{rot}^{rel}=\Gamma e_{rot}(T_{eq}) 
+  (1-\Gamma) \left[ \theta e_{rot}(T_{tr,rot}) + (1-\theta) E_{rot}(f)
 \right]
,\\
& e_{vib}^{rel}= \Gamma e_{vib}(T_{eq})+ (1-\Gamma) E_{vib}(f),
\end{split}
\end{equation}
and the corresponding relaxation temperatures are
\begin{equation}  \label{eq-fermeture2}
T_{rot}^{rel}=e_{rot}^{-1}(e_{rot}^{rel}), \quad \text{ and } 
\quad T_{vib}^{rel}=e_{vib}^{-1}(e_{vib}^{rel}).
\end{equation}
These definitions account for the relaxation of $T_{rot}$ to
$T_{tr,rot}$ then to $T_{eq}$, and for the relaxation of $T_{vib}$ to
$T_{eq}$ with rates that are consistent with the definition of $\Pi$.

Note that the relaxation rotational temperature $T_{rot}^{rel}$
can be equivalently defined by 
$T_{rot}^{rel} =   \Gamma T_{eq} + (1-\Gamma) \left[ \theta T_{tr,rot} +
  (1-\theta) T_{rot} \right] $, which is a simple extension of the
definition given in~\cite{esbgk_poly}. However, the relaxation
vibrational temperature $T_{vib}^{rel}$  cannot be defined in the same
way: indeed, the nonlinearity of the function $e_{vib}$ would make the
simpler definition $T_{vib}^{rel}= \Gamma T_{eq} + (1-\Gamma) T_{vib}$
not consistent with the energy conservation (see section~\ref{subsec:conservation}).

This derivation shows that parameter $\theta$ is
associated with transfers between translational and rotational
energies and $\Gamma$ with transfers between
translational-rotational and vibrational energies. It will be shown in
section~\ref{subsec:relax} that these parameters are related to
$Z_{rot}$ and $Z_{vib}$ by the relations
\begin{equation}\label{eq-defGammatheta} 
  \Gamma=\frac{1}{Z_{vib}},\quad \text{ and } \quad
\theta=\frac{1/Z_{rot}-1/Z_{vib}}{1-1/Z_{vib}}.
\end{equation}
Moreover, parameter $\nu$ will be used to fit the correct Prandtl
number. It will be shown in section~\ref{sec:relaxf} that $\nu$ has to be set so that the Prandtl number $Pr$ is
\begin{equation}\label{eq-Pr} 
  Pr=\frac{1}{1-(1-\Gamma)(1-\theta)\nu}.
\end{equation}
Finally, the relaxation time $\tau$ of the model is
\begin{equation*}
  \tau=\frac{\mu}{p}(1-(1-\Gamma)(1-\theta)\nu),
\end{equation*}
as it is proved in section~\ref{sec: NS}.

\subsection{Conservation properties}
\label{subsec:conservation}

For the analysis of the conservation properties of our model, it is
useful to define the relaxation energy of translation
\begin{equation*}
  e_{tr}^{rel} = \frac12 \Tr(\Pi) 
= \Gamma \frac32 R T_{eq}  + (1-\Gamma) ( \theta \frac32 R
T_{tr,rot} + (1-\theta)\frac32 R T_{tr}),
\end{equation*}
and the corresponding relaxation temperature of translation which is
\begin{equation}  \label{eq-Ttrrel}
T_{tr}^{rel}=e_{tr}^{-1}(e_{tr}^{rel}),
\end{equation}
that will be used later.

Then, note that this relation and the definition~(\ref{eq: fermeture})
of relaxation energies of rotation and vibration can be rewritten
under the compact form
\begin{equation} \label{conserv}
\begin{pmatrix} e_{tr}^{rel}\\ e_{rot}^{rel}\\ e_{vib}^{rel}\end{pmatrix}
=\Gamma
\begin{pmatrix} e_{tr}(T_{eq})\\ e_{rot}(T_{eq})\\ e_{vib}(T_{eq})\end{pmatrix}+(1-\Gamma)\begin{pmatrix} 1-\frac{\delta\theta}{3+\delta} & \frac{3\theta}{3+\delta} & 0\\ \frac{\delta\theta}{3+\delta} & 1-\frac{3\theta}{3+\delta} & 0 \\ 0& 0& 1\end{pmatrix}
\begin{pmatrix} E_{tr}(f)\\ E_{rot}(f)\\ E_{vib}(f)\end{pmatrix}.
\end{equation}
Now, we state what are the first moments of ${\cal G}[f]$ that can be
computed by standard integrals and series (see appendix~\ref{app:integrals}).
\begin{proposition} \label{prop:cons}
  The Gaussian  ${\cal G}[f]$ satisfies
\begin{align}
&  \la {\cal G}[f] \ra_{v,\varepsilon,i} = \rho,  \label{eq-Grho}\\
&  \la v{\cal G}[f] \ra_{v,\varepsilon,i} = \rho u,  \label{eq-Grhou}\\
& \la \frac12 |v-u|^2{\cal G}[f] \ra_{v,\varepsilon,i} = \rho  e_{tr}^{rel}, \quad 
\la \varepsilon {\cal G}[f] \ra_{v,\varepsilon,i} = \rho e_{rot}^{rel}, \quad
\la iRT_0 {\cal G}[f] \ra_{v,\varepsilon,i} = \rho e_{vib}^{rel}. \label{eq-Gint}
\end{align}
\end{proposition}

Then these properties can be used to prove the conservations properties
of our kinetic model.
\begin{proposition}
The collision operator of the ES-BGK model satisfies the conservation
of mass, momentum, and energy:
\begin{equation*}
  \la (1,v,\frac12 |v-u|^2+ \varepsilon + iRT_0)
\frac{1}{\tau}({\cal G}[f]-f) \ra_{v,\varepsilon,i}  = 0.
\end{equation*}
\end{proposition}
\begin{proof}
  The conservation of mass and momentum are obvious consequences of
  relations~\eqref{eq-Grho} and~\eqref{eq-Grhou}. For the conservation
  of energy, note that~\eqref{eq-Gint} and~\eqref{conserv} imply
\begin{equation*}
\begin{split}
   \la (\frac12 |v-u|^2+ \varepsilon + iRT_0) {\cal G}[f]
  \ra_{v,\varepsilon,i}  
& = \rho ( e_{tr}^{rel}+ e_{rot}^{rel}+ e_{vib}^{rel}) \\
& = \rho  \Gamma( e_{tr}(T_{eq})+ e_{rot}(T_{eq})+ e_{vib}(T_{eq})) 
+ \rho(1-\Gamma) (E_{tr}(f) + E_{rot}(f) + E_{int}(f)) \\
&  = \rho \Gamma E(f) + \rho (1-\Gamma)E(f) \\
&  = \rho E(f) =  \la (\frac12 |v-u|^2+ \varepsilon + iRT_0) f
  \ra_{v,\varepsilon,i}, 
\end{split}
\end{equation*}
where we have used relations~\eqref{eq-E_T}--\eqref{eq-ETeq} and \eqref{eq-mtsf}.
\end{proof}

\subsection{\label{sec: entropy}Entropy}

Andries et al. \cite{esbgk_poly} proved that the ES-BGK model for
polyatomic gases satisfies the entropy dissipation property. Since our
ES-BGK model is an extension to include the energy of vibration of
polyatomic gases we follow the same proof. First, the rotational
energy variable $\varepsilon$ is transformed to the variable $I$ so
that $\varepsilon=I^{2/\delta}$. With this new variable, the
distribution function of the gas now is $g(t,x,v,I,i)$, defined such
that $g(t,x,v,I,i)dI=f(t,x,v,\varepsilon,i)d\varepsilon$, which gives
\begin{equation*}
g(t,x,v,I,i)=\frac{2}{\delta}\varepsilon^{1-\delta/2}f(t,x,v,\varepsilon,i).
\end{equation*}
Our ES-BGK model given by~\eqref{eq: cinetique} and~\eqref{eq:
  operator} now reads
\begin{equation}\label{eq-ESBGKI} 
  \partial_t g+ v\cdot\nabla g=\frac{1}{\tau}(\mathcal{G}[g]-g),
\end{equation}
where the Gaussian distribution now reads:
\begin{equation}\label{eq-GI} 
\begin{split}
\mathcal{G}[g]&
=\rho\frac{2}{\delta}\frac{\Lambda(\delta)
  (1-\exp(-T_0/T_{vib}^{rel})) }{\sqrt{\det(2\pi\Pi)}
  (RT_{rot}^{rel})^{\delta/2} }\exp\left(
  -\frac{1}{2}(v-u)^{T}\,\Pi^{-1}\,(v-u)
  -\frac{I^{2/\delta}}{RT_{rot}^{rel}} -i\frac{T_0}{T_{vib}^{rel}}
\right)\\
& =\exp(\boldsymbol{\alpha}^T\mathbf{m}),
\end{split}
\end{equation}
with
\begin{align*}
& \Lambda(\delta)=\frac{\delta}{2}\left( \int_{\mathbb{R}}\exp(-I^{2/\delta})dI \right)^{-1},\\
& \mathbf{m}=(1,v,v \otimes v,I^{2/\delta}, iRT_0)^T,\\
& \boldsymbol{\alpha}=\left(\log\left( \frac{2}{\delta}\frac{\rho \Lambda(\delta)(1-\exp(-T_0/T_{vib}^{rel}))}{\sqrt{\det (2\pi\Pi)}(RT_{rot}^{rel})^{\delta/2}} \right)-\frac{1}{2}u^T\Pi^{-1}u,\Pi^{-1}u,-\frac{1}{2}\Pi^{-1},-\frac{1}{RT_{rot}^{rel}},-\frac{1}{RT_{vib}^{rel}}\right).
\end{align*}
The corresponding Maxwellian equilibrium now is
\begin{equation*}
  {\cal M}[g] = \rho\frac{2}{\delta}\frac{\Lambda(\delta)
  (1-\exp(-T_0/T_{eq})) }{ (2\Pi R T_{eq})^{3/2} (RT_{eq})^{\delta/2} }\exp\left(
  -\frac{|v-u|^2}{2RT_{eq}}
  -\frac{I^{2/\delta}}{RT_{eq}} -i\frac{T_0}{T_{eq}}
\right).
\end{equation*}
This transformation makes all the proofs of this section much simpler.
Now we give the conditions on which our model is well defined, and we
state its entropy property.
\begin{proposition} \label{prop:theoH}
For parameters $-1/2
\leq \nu < 1$, $0\leq\theta<1$, and $0\leq\Gamma<1$ we have:
\begin{enumerate}
\item For symmetric positive definite tensor $\Theta$ and positive
  temperatures $T_{tr,rot}$ and $T_{tr}$, the tensor $\Pi$ defined by~\eqref{eq: Pi} is symmetric positive definite.
\item  (Entropy minimization) If $g$ is a non-negative distribution,
  then the
  Gaussian distribution $\mathcal{G}[g]$ defined by~\eqref{eq-GI} is the unique
   minimizer of the entropy $\ent(g) = \la g\log g -g\ra_{v,I,i}$ on the
   set  ${\cal X}=\lbrace \phi\geq 0,\ \la
     \mathbf{m}\phi\ra_{v,I,i}
 =\left(\rho,\rho u,\rho(u\otimes u+\Pi),\rho e_{rot}^{rel},\rho e_{vib}^{rel} \right)\rbrace.$
\item  (H-theorem) The ES-BGK model~\eqref{eq-ESBGKI} satisfies
\begin{equation*}
  \partial_t \ent(g)+ \nabla \cdot  \la v ( g\log g -g)\ra_{v,I,i}=\la
  \frac{1}{\tau}(\mathcal{G}[g]-g)\log g\ra_{v,I,i} \leq 0,
\end{equation*}
\item  (Equilibrium) If $g=\mathcal{G}[g]$, then $g=\mathcal{M}[g]$.
\end{enumerate}
\end{proposition}



\begin{proof}[Proof of Property 1.]
We first rewrite $\Pi$ as follows: we define the intermediate stress
tensor $A=\nu\Theta+(1-\nu)RT_{tr}I$ associated to the relaxation
phenomenon for the translation mode, and the tensor $B=(1-\theta)A+\theta R
T_{tr,rot}I$ associated to the relaxation of the rotational mode, such
that~\eqref{eq: Pi} reads $\Pi=(1-\Gamma)B+\Gamma RT_{eq}I$. Andries
et al.~\cite{esbgk_poly} have proved that tensor $A$ is positive definite for
  $\nu \in [-1/2,1]$. Now, for $\theta\in[0,1]$, since $B$ is a convex
  combination of $A$ and $R T_{tr,rot}I$, it is also symmetric and
  positive definite. Finally, for $\Gamma \in [0,1]$, $\Pi $ is a
  convex combination of $B$ and $RT_{eq}I$, and hence is symmetric and
  positive definite too.
\end{proof}

\begin{proof}
[Proof of Property 2.]
First, note that by construction, ${\cal G}[g]$ is in set ${\cal
  X}$. Then, since the functional $g\mapsto \ent(g)$ is convex, then we have 
\begin{equation*}
\ent(\mathcal{G}[g]) \leq \ent(\phi) -\ent'(\mathcal{G}[g])(\phi-\mathcal{G}[g])
\end{equation*}
for every $\phi$ in ${\cal X}$.
Moreover, we have 
\begin{equation*}
\begin{split}
\ent'(\mathcal{G}[g])(\phi-\mathcal{G}[g])  
& = \la  (\phi-\mathcal{G}[g] ) \log\mathcal{G}[g]  \ra_{v,I,i} \\
& = \la  (\phi-\mathcal{G}[g] ) \boldsymbol{\alpha}^T\mathbf{m}
\ra_{v,I,i} \\
& = 0,
\end{split}
\end{equation*}
since both $\mathcal{G}[g]$ and $\phi$ are in ${\cal
  X}$. Consequently $\ent(\mathcal{G}[g]) \leq \ent(\phi)$
for every $\phi$ in ${\cal X}$, which concludes the proof.




\end{proof}

\begin{proof}[Proof of property 3.]

This proof is decomposed into 4 steps.

\paragraph{Step 1: entropy inequality.}
First, note that with elementary calculus,~\eqref{eq-ESBGKI} implies
\begin{equation*}
  \partial_t \ent(g)+ \nabla \cdot  \la v ( g\log g -g)\ra_{v,I,i}
= \frac{1}{\tau}\ent'(g)(\mathcal{G}[g]-g).
\end{equation*}
Then, since $\ent$ is convex, the right-hand side of the previous
equality satisfies
\begin{equation*}
\ent'(g)(\mathcal{G}[g]-g) \leq \ent(\mathcal{G}[g]) - \ent(g).  
\end{equation*}
Consequently, the H-theorem is obtained if we can prove that 
\begin{equation}\label{eq-minGg} 
  \ent(\mathcal{G}[g]) \leq \ent(g).
\end{equation}
Note that this is not obvious, since $g$ is not in ${\cal X}$. 

\paragraph{Step 2: entropy minima on different sets.}
It is convenient to
define, for every macroscopic quantities $\rho$, $u$, $\Pi$, $T_{rot}^{rel}$ and
$T_{vib}^{rel}$ the minimum of entropy $\ent$ on $\cal X$, and we set
\begin{equation*}
  S(\rho,u,\Pi,T_{rot}^{rel},T_{vib}^{rel}) 
= \min\left\lbrace \ent(\phi), \phi\geq 0 \text{ s.t.}  \la \mathbf{m}\phi\ra_{v,I,i}
 =\left(\rho,\rho u,\rho(u\otimes u+\Pi),\rho e_{rot}^{rel},\rho e_{vib}^{rel} \right)\right\rbrace.
\end{equation*}
Property 2 implies 
\begin{equation*}
S(\rho,u,\Pi,T_{rot}^{rel},T_{vib}^{rel})
= \ent(\mathcal{G}[g]).  
\end{equation*}

Now we define a second entropy minimization problem, based on the
moments of $g$. Namely
\begin{equation*}
  S(\rho,u,\Theta,T_{rot},T_{vib}) 
= \min\left\lbrace \ent(\phi), \phi\geq 0 \text{ s.t.}  \la \mathbf{m}\phi\ra_{v,I,i}
 =\left(\rho,\rho u,\rho(u\otimes u+\Theta),\rho E_{rot}(f),\rho E_{vib}(f) \right)\right\rbrace.
\end{equation*}
Here, by definition $g$ belongs to the minimization set, and therefore
 \begin{equation*}
S(\rho,u,\Theta,T_{rot},T_{vib}) \leq  \ent(g) .
 \end{equation*}

Therefore, a sufficient condition to have~\eqref{eq-minGg} is 
$S(\rho,u,\Pi,T_{rot}^{rel},T_{vib}^{rel}) \leq
S(\rho,u,\Theta,T_{rot},T_{vib}) $, which is rewritten as
\begin{equation}  \label{eq-DS}
\Delta S = S(\rho,u,\Pi,T_{rot}^{rel},T_{vib}^{rel})  -
S(\rho,u,\Theta,T_{rot},T_{vib}) \leq 0.
\end{equation}
 This entropy difference is now analyzed in the following.

\paragraph{Step 3: entropy difference}
A direct calculation gives
\begin{equation*}
  S(\rho,u,\Pi,T_{rot}^{rel},T_{vib}^{rel})
= \rho\log\left( \rho\frac{2}{\delta}\frac{\Lambda(\delta)(1-\exp(-T_0/T_{vib}^{rel}))}{\sqrt{\det (2\pi\Pi)}(RT_{rot}^{rel})^{\delta/2}}\right)-\rho\frac{5+\delta+\delta_v(T_{vib}^{rel})}{2}.
\end{equation*}

A similar relation is deduced for
$S(\rho,u,\Theta,T_{rot},T_{vib}) $ and we get
\begin{align*}
\Delta S
& =\frac{1}{2}\rho\log\left( \frac{\det \Theta}{\det\Pi}\left(\frac{T_{rot}}{T_{rot}^{rel}}\right)^\delta\left(\frac{(1-\exp(-T_0/T_{vib}^{rel}))}{(1-\exp(-T_0/T_{vib}))}\right)^2 \right)-\rho\frac{\delta_v(T_{vib}^{rel})-\delta_v(T_{vib})}{2},\\
&=\frac{1}{2}\rho\log\left( \frac{\det \Theta}{\det\Pi}\left(\frac{E_{rot}(f)}{e_{rot}^{rel}}\right)^\delta\left(\frac{RT_0+E_{vib}(f)}{RT_0+e_{vib}^{rel}}\right)^2 \right)-\rho\frac{\delta_v(T_{vib}^{rel})-\delta_v(T_{vib})}{2}
\end{align*}
where we have used relations~\eqref{eq-evib},~\eqref{eq-T_tr_rot_vib},
and \eqref{eq-fermeture2} to obtain the last equality.

First, the following result is admitted (see the proof in appendix \ref{subsec:ineqthetapi}):
\begin{eqnarray} \label{ineqthetapi}
\frac{\det \Theta}{\det\Pi}\leq \left(\frac{E_{tr}(f)}{e_{tr}^{rel}}\right)^3.
\end{eqnarray}
This allows us to write the following inequality, as function of energies only:
\begin{equation*}
  \Delta S
\leq \frac{1}{2}\rho\log\left( \left(\frac{E_{tr}(f)}{e_{tr}^{rel}}\right)^3 \left(\frac{E_{rot}(f)}{e_{rot}^{rel}}\right)^\delta\left(\frac{RT_0+E_{vib}(f)}{RT_0+e_{vib}^{rel}}\right)^2 \right)-\rho\frac{\delta_v(T_{vib}^{rel})-\delta_v(T_{vib})}{2}.
\end{equation*}
After expansion, this inequality reads as 
\begin{equation}  \label{eq-DSS}
\Delta S \leq \frac{\rho}{R}\left(\mathcal{S}(E_{tr}(f),E_{rot}(f),E_{vib}(f)) - \mathcal{S}(e^{rel}_{tr},e^{rel}_{rot},e^{rel}_{vib})\right),
\end{equation}
where we have introduced the new energy functional $\mathcal{S}$,
defined for every energy triplet $(e_1,e_2,e_3)$ by
\begin{equation*}
\mathcal{S}(e_1,e_2,e_3)=R\left(\frac32\log(e_1)+\frac{\delta}2\log(e_2)+\log\left(1+\frac{e_3}{RT_0}\right)+\frac{e_3}{RT_0}\log\left(1+\frac{RT_0}{e_3}\right)\right).
\end{equation*}
Note that to obtain~(\ref{eq-DSS}), we also have replaced $\delta_v$ by its
definition~\eqref{eq-deltav} and the temperatures of
vibration have been replaced by their corresponding energies.

Now it is clear that a sufficient condition to have $\Delta S \leq 0$
is 
\begin{equation}  \label{eq-SmS}
\mathcal{S}(E_{tr}(f),E_{rot}(f),E_{vib}(f)) \leq \mathcal{S}(e^{rel}_{tr},e^{rel}_{rot},e^{rel}_{vib}),
\end{equation}
which is proved in the last step.

\paragraph{Step 4: proof of (\ref{eq-SmS})}

The usual argument to conclude an entropy inequality is a convexity
property. Here, our functional $\mathcal{S}$ can easily be seen to be
concave (see appendix~\ref{sec:appS}). However, since the right-hand
side of \eqref{eq-SmS} is not at equilibrium, a direct use of the
convexity inequality does not work here. Instead, we find it
simpler, and physically relevant, to use successively two paths, based
on parameters $\theta$ and $\Gamma$. Indeed, note that relaxation
energies $(e^{rel}_{tr},e^{rel}_{rot},e^{rel}_{vib})$ depend on
$\theta$ and $\Gamma$ (see \eqref{conserv}). Then we set
\begin{equation*}
  s(\theta,\Gamma) = \mathcal{S}(e^{rel}_{tr},e^{rel}_{rot},e^{rel}_{vib}).
\end{equation*}
From \eqref{conserv}, it is clear that $s(0,0) =
\mathcal{S}(E_{tr}(f),E_{rot}(f),E_{vib}(f))$ since the relaxation
energies reduce to the internal energies of $f$ for such values of
$\theta$ and $\Gamma$. Consequently, inequality \eqref{eq-SmS} reduces
to
\begin{equation}  \label{eq-s}
s(0,0) \leq s(\theta,\Gamma).
\end{equation}
The idea is now to decompose inequality \eqref{eq-s}
into two embedded inequalities
\begin{equation}\label{eq-ss} 
s(0,0) \leq s(\theta,0) \leq s(\theta,\Gamma).
\end{equation} 

We start with the second inequality and consider the variation of $s$
with respect to $\Gamma$. Elementary calculus shows that
\begin{equation}\label{eq-dGammas} 
\begin{split}
  \frac{\partial s}{\partial \Gamma}(\theta,\Gamma) & = \frac{1}{T_{tr}^{rel}}\left(e_{tr}(T_{eq})-\left(1-\frac{\delta\theta}{3+\delta}\right)E_{tr}(f)-\frac{3\theta}{3+\delta}E_{rot}(f)\right)\\
&+\frac{1}{T_{rot}^{rel}}\left(e_{rot}(T_{eq})-\left(1-\frac{3\theta}{3+\delta}\right)E_{rot}(f)-\frac{\delta\theta}{3+\delta}E_{tr}(f)\right)\\
&+\frac{1}{T_{vib}^{rel}}\left(e_{vib}(T_{eq})-E_{vib}(f)\right),
\end{split}
\end{equation}
and
\begin{equation*}
\begin{split}
\frac{\partial^2 s}{\partial \Gamma^2}(\theta,\Gamma) & 
= \partial_{1,1}\mathcal{S}(e_{tr}^{rel},e_{rot}^{rel},e_{vib}^{rel})
        \left(e_{tr}(T_{eq})-\left(1-\frac{\delta\theta}{3+\delta}\right)E_{tr}(f)-\frac{3\theta}{3+\delta}E_{rot}(f)\right)^2\\
&+\partial_{2,2}\mathcal{S}(e_{tr}^{rel},e_{rot}^{rel},e_{vib}^{rel})\left(e_{rot}(T_{eq})-\left(1-\frac{3\theta}{3+\delta}\right)E_{rot}(f)-\frac{\delta\theta}{3+\delta}E_{tr}(f)\right)^2\\
&+\partial_{3,3}\mathcal{S}(e_{tr}^{rel},e_{rot}^{rel},e_{vib}^{rel})\left(e_{vib}(T_{eq})-E_{vib}(f)\right)^2,
\end{split}
\end{equation*}
and the reader is referred to appendix~\ref{sec:appS} for the computation of the
partial derivatives of $\mathcal{S}$.
The previous relation shows that $s$ is a concave function of
$\Gamma$. Moreover, note that for $\Gamma=1$, relation~\eqref{conserv}
shows that all the relaxation energies are equal to the equilibrium
energy, and hence all the relaxation temperatures are equal to
$T_{eq}$. When this is used into~(\ref{eq-dGammas}), we find that
$\frac{\partial s}{\partial \Gamma}(\theta,1) = 0$. With the concavity
property, this proves that $s$ is an increasing function of $\Gamma$
on the interval $[0,1]$, and this proves the second inequality
of~(\ref{eq-ss}).

For the first inequality of~(\ref{eq-ss}), we set $\Gamma $ to 0, and
we study the variation of $s(\theta,0)$ with respect to $\theta$.  Again, elementary
calculus shows that $\frac{\partial^2 s}{\partial
  \theta^2}(\theta,0)\leq 0$, and hence $s(\theta,0) $ is a concave
function of $\theta$. Moreover, we find
\begin{equation}\label{eq-dthetas} 
\begin{split}
\frac{\partial s}{\partial \theta} (0,0) =
-\frac{3\delta}{2(3+\delta)}R (T_{tr}-T_{rot})(\frac{1}{T_{tr}} - \frac{1}{T_{rot}}),
\end{split}
\end{equation}
which implies that $s(\theta,0) $ is a non
decreasing function of $\theta$. Consequently, this gives the first
inequality of~(\ref{eq-ss}) which concludes the proof of~\eqref{eq-SmS},
and hence of~\eqref{eq-minGg}, and the proof of the H-theorem is now complete.

\end{proof}

\begin{proof}[Proof of property 4.]
  At equilibrium $g=\mathcal{G}[g]$ and hence $\Theta = \Pi$,
  $E_{rot}(g) = e_{rot}^{rel}$, and $E_{vib}(g) = e_{vib}^{rel}$. Then
  it is easy to see that relations~\eqref{eq:
    Pi}--\eqref{eq-fermeture2} imply
  $ T_{tr}=T_{rot}=T_{vib} =T_{tr,rot} = T_{rot}^{rel} = T_{vib}^{rel}
  =T_{eq}$
  and then $\Theta = RT_{eq} I$. Consequently,
  $\mathcal{G}[g] =\mathcal{M}[g]$ and then $g =\mathcal{M}[g]$.
\end{proof}

\begin{remark}
Of course, the equivalent H-theorem for our initial model (with
function $f$ and variable $\varepsilon$) can then be obtained by using the change
of variable $\varepsilon=I^{\frac{2}{\delta}}$. However, note that the
entropy functional now reads $\ent(f) = \langle f\log(f/\varepsilon^{\frac{\delta}2-1} )\rangle_{v,\varepsilon,i}$.
\end{remark}

\section{Relaxation phenomena \label{sec: energ}}

In this section, we resolve the local relaxation equations for
energies, stress tensor, and heat flux. This give us the relations
between parameters $\Gamma$, $\theta$, and $\nu$ of our model and the
vibrational and rotation collision numbers $Z_{vib}$, $Z_{rot}$, and
the Prandtl number. 

\subsection{Relaxation rates of translational, rotational and
  vibrational energies}
\label{subsec:relax}

The energy of translation, rotation and vibration are transferred from
one mode to another one during inter-molecular collisions. These transfers
are described by local relaxations obtained as moments of our ES-BGK
model (in a space homogeneous case). Indeed, our model~\eqref{eq:
  cinetique}--\eqref{eq: operator} is multiplied by
$\frac{1}{2}|v-u|^2$, $\varepsilon$, $iRT_0$, and integrated w.r.t $v$,
$\varepsilon$, and $i$, and we use closure relations~\eqref{conserv} to find
\begin{equation} \label{eq:edo}
\frac{d}{dt}
 \begin{pmatrix} E_{tr}(f)\\ E_{rot}(f)\\ E_{vib}(f)\end{pmatrix}=\frac{\Gamma}{\tau}
\begin{pmatrix} e_{tr}(T_{eq})-E_{tr}(f)\\ 
               e_{rot}(T_{eq})-E_{rot}(f)\\
               e_{vib}(T_{eq})-E_{vib}(f)
\end{pmatrix}
+\frac{1-\Gamma}{\tau}
\begin{pmatrix} 
-\frac{\delta\theta}{3+\delta} & \frac{3\theta}{3+\delta} & 0 \\ 
\frac{\delta\theta}{3+\delta} & -\frac{3\theta}{3+\delta} & 0 \\ 
0& 0& 0
\end{pmatrix}
\begin{pmatrix} E_{tr}(f)\\ E_{rot}(f)\\ E_{vib}(f)\end{pmatrix}
\end{equation}
The last equation has to be consistent with the Landau-Teller 
relaxation equation that describes the relaxation of the macroscopic energy of vibration
to equilibrium, at a relaxation rate
$\tau_{vib}=\tau Z_{vib}$. The second equation has to be consistent
with the Jeans relaxation equation, which plays the same role for
rotational energy, at the rate $\tau_{rot}=\tau Z_{rot}$. Moreover,
this equation should also be consistent with the fast relaxation of
$T_{tr}$ and $T_{rot}$ towards $T_{tr,rot}$ (see
section~\ref{sec:relax}).

Now we assume parameters $\tau$, $\Gamma$, and $\theta$ to be
constant, and we solve these equations to find
\begin{equation*}
\begin{aligned}
E_{vib}(f(t))& = e_{vib}(T_{eq}) 
                + (E_{vib}(f(0))-e_{vib}(T_{eq})) \exp \left(-\frac{\Gamma}{\tau}t\right),\\
E_{rot}(f(t))& = e_{rot}(T_{eq}) \\
& \quad  +\left[E_{rot}(f(0))-e_{rot}(T_{eq})+\frac{\delta}{3+\delta}(E_{vib}(f(0))-e_{vib}(T_{eq}))\right]
\exp\left(-\frac{1-(1-\Gamma)(1-\theta)}{\tau} t \right)\\
&\quad
-\frac{\delta}{3+\delta}(E_{vib}(f(0))-e_{vib}(T_{eq}))\exp\left(-\frac{\Gamma}{\tau}
t\right).
\end{aligned}
\end{equation*}
From these equations we deduce that:
\begin{equation*}
Z_{vib}=\frac{1}{\Gamma},\quad 
Z_{rot}=\frac{1}{1-(1-\Gamma)(1-\theta)},
\end{equation*}
or equivalently $\Gamma=1/Z_{vib}$ and $\theta={(Z_{vib}-Z_{rot})}/{((Z_{vib}-1)Z_{rot})}$.

Since we want the rotational and vibrational collision numbers such
that $1<Z_{rot}<Z_{vib}$ (see section~\ref{sec:relax}), then the
previous definition gives the restriction $0 \leq \theta <1$ and $0
\leq \Gamma <1$. Case $\theta=0$ gives $Z_{rot}=Z_{vib}$ which means
that vibration modes relax as fast as rotation modes. In case
$\Gamma=0$, then $Z_{vib}=+\infty$ and $Z_{rot}=1/\theta$, and  we
find the polyatomic ES-BGK model without vibrations of Andries et al.~\cite{esbgk_poly}.

The equivalent relaxations of temperatures are
\begin{equation}
\label{eq:diff_temp}
\begin{aligned}
R(T_{tr,rot}-T_{eq})=&-\frac{2}{3+\delta}(E_{vib}(f(0))-e_{vib}(T_{eq}))\exp\left(-\frac{t}{\tau Z_{vib}}\right)\\
R(T_{tr}-T_{tr,rot})=&-\frac{2}{3}\left[E_{rot}(f(0))-e_{rot}(T_{eq})+\frac{\delta}{3+\delta}(E_{vib}(f(0))-e_{vib}(T_{eq}))\right]\exp\left(-\frac{t}{\tau Z_{rot}}\right)
\end{aligned}
\end{equation}
These two expressions will be used in the numerical tests of
section~\ref{sec: numeric} to check the correct rates of convergence
to equilibrium.

\subsection{\label{sec:relaxf}Relaxation of stress and heat flux}

Relaxation equations for stress tensor and heat flux are obtained by
multiplying the kinetic equation~(\ref{eq: cinetique}) by
$(v-u)\otimes(v-u)$ and $(\frac{1}{2}|v-u|^2+\varepsilon+iRT_0)(v-u)$ and
integrating w.r.t $v$,$\varepsilon$ and $i$ to get, in the space
homogeneous case :
\begin{align}
&\frac{d}{dt} \Theta=\frac{1}{\tau}(
(1-\Gamma)(1-\theta)(1-\nu)(RT_{tr}I-\Theta)+(1-\Gamma)\theta(RT_{tr,rot}I-\Theta)+\Gamma(RT_{eq}I-\Theta)),
\label{eq-dTheta} \\
&\frac{d}{dt} q=-\frac{1}{\tau}q \label{eq-dq}   .
\end{align}
Since $\Tr(\Theta) = 3RT_{tr}$, taking the trace of~(\ref{eq-dTheta})
gives
\begin{equation*}
 \frac{d}{dt}RT_{tr}=\frac{1}{\tau}\left( (1-\Gamma)\theta(RT_{tr,rot}-RT_{tr})+\Gamma(RT_{eq}-RT_{tr}) \right).
\end{equation*}
This equation is subtracted to~(\ref{eq-dTheta}) to get
\begin{equation*}
  \frac{d}{dt} (\Theta - RT_{tr}I) =
  -\frac{1}{\tau}(1-(1-\theta)(1-\Gamma)\nu) (\Theta - RT_{tr}I).
\end{equation*}
This shows that for large times, the stress tensor tends to $RT_{tr}I$, while the heat
flux tends to 0.
More precisely, for $\nu$, $\theta$, $\Gamma$ and $\tau$ constant, we
have the analytic solutions:
\begin{align*}
  & \Theta(t)-RT_{tr}(t)I = (\Theta(0)-RT_{tr}(0)I)
    \exp\left(-(1-(1-\Gamma)(1-\theta)\nu)\frac{t}{\tau} \right), \\
 & q(t)=  q(0)\exp\left(-\frac{t}{\tau}\right).
\end{align*}
The Prandtl number can be viewed as the ratio between the relaxation times of these two processes, and we get:
\begin{equation*}
  \Pr = \frac{1}{1-(1-\Gamma)(1-\theta)\nu}.
\end{equation*}
Incidentally, this value will be checked numerically in
section~\ref{sec: numeric} by computing the ratio
\begin{equation}\label{eq-ratioqT} 
\frac{\log(|q_i(t)/q_i(0)|)}{\log(|(\Theta_{ii}(t)-RT_{tr}(t))/(\Theta_{ii}(0)-RT_{tr}(0))|)}
\end{equation}
for $i=1,2,3$.

\section{Chapman-Enskog analysis \label{sec: chap}}
The conservation laws are obtained by multiplying~(\ref{eq:
  cinetique}) by the vector $1$, $v$, and $\frac{1}{2}|v|^2$ and then
by integrating it to get:
\begin{equation}
\label{eq: conservation}
\begin{aligned}
&\partial_t \rho+\nabla \cdot (\rho u)=0,\\
&\partial_t (\rho u)+\nabla \cdot (\rho u \otimes u)+ \nabla\cdot\Sigma(f)=0,\\
&\partial_t \mathcal{E}+\nabla \cdot (\mathcal{E}u)+ \nabla \cdot
\Sigma(f) + \nabla \cdot q(f)=0,
\end{aligned}
\end{equation}
where $\mathcal{E} =   \langle(
\frac{1}{2}|v|^2+\varepsilon+iRT_0)f\rangle_{v,\varepsilon,i}=\frac{1}{2}\rho|u|^2+\rho E(f) $ is the
total energy density, $\Sigma(f) = \langle (v-u) \otimes (v-u)
f\rangle_{v,\varepsilon,i} = \rho \Theta$ is the stress tensor and $q(f) = \langle (\frac{1}{2}|v-u|^2+\varepsilon+iRT_0) (v-u) f\rangle_{v,\varepsilon,i}$ is the
heat flux.

If we have some characteristic values of length, time,
velocity, density, and temperature, our ES-BGK model~\eqref{eq:
  cinetique}--\eqref{eq: operator} can be non-dimensionalized. This
equation reads
\begin{equation}
\label{eq: nd_kinetic}
\partial_t f+v\cdot \nabla f=\frac{1}{\Kn \, \tau }(\mathcal{G}[f]-f),
\end{equation}
where $\Kn$ is the Knudsen number which is the ratio between the mean
free path and a macroscopic length scale. For simplicity, here we use the
same notations for the non-dimensional variables as for the dimensional ones.

The Chapman-Enskog analysis consists in approximating the stress
tensor and the heat flux at first and second order with
respect to the Knudsen number, which gives compressible Euler
equations and compressible Navier-Stokes equations, respectively.

\subsection{Euler asymptotics}

At equilibrium $f$, is equal to the equilibrium Maxwellian
distribution. Even in non-equilibrium, when $\Kn$ is very small the gas is very close to its
equilibrium state, and equation~\eqref{eq: cinetique}--\eqref{eq:
  operator} gives
\begin{equation}  \label{eq: fequilibre}
f=\mathcal{M}[f]+O(\Kn),
\end{equation}
if in addition $f$ and its time and space derivatives are $O(1)$ w.r.t
$\Kn$. Then
definition~\eqref{eq-Theta_q} gives
\begin{equation}
\label{eq: fequilibre2}
\Sigma(f)=pI+O(\Kn), \qquad q(f)=O(\Kn),
\end{equation}
where we denote by $p=\rho R T_{eq}$ the pressure at equilibrium.

These last relations are used into conservation laws~(\ref{eq:
  conservation}) to get the compressible Euler equations with first
order reminder:
\begin{equation}\label{eq-euler} 
\begin{aligned}
&\partial_t \rho+\nabla \cdot(\rho u)=0,\\
&\partial_t (\rho u)+\nabla \cdot(\rho u \otimes u)+\nabla p=O(\Kn),\\
&\partial_t \mathcal{E}+\nabla \cdot\left((\mathcal{E}+p)u\right)=O(\Kn).\\
\end{aligned}
\end{equation}
The non-conservative form of these equations is
\begin{equation}\label{eq-euler_nc} 
\begin{aligned}
&\partial_t \rho+u\cdot \nabla\rho+\rho\nabla \cdot u=0,\\
&\partial_t u+(u\cdot \nabla)u+\frac{1}{\rho}\nabla p=O(\Kn),\\
&\partial_t T_{eq}+u\cdot \nabla T_{eq}+T_{eq} C \nabla \cdot u=O(\Kn),
\end{aligned}
\end{equation}
with $ C=\frac{R}{c_v(T_{eq})}$, and $c_v(T_{eq}) = \frac{\partial
  e(T_{eq})}{\partial T_{eq}}$ is the heat capacity at constant volume
  of the gas, which is temperature dependent here due to vibration
  modes (see equations~\eqref{eq-Teq} and~\eqref{eq-ETeq}).

Moreover, simple calculations give $E_{\alpha}(\mathcal{M}[f]) =
e_{\alpha}(T_{eq})$ for $\alpha=tr, rot, vib$. Since the energy
functions are regular, our expansion $f=\mathcal{M}[f]+O(\Kn)$ and
relations~\eqref{eq-T_tr_rot_vib} and~\eqref{eq-T_trrot}  give
\begin{equation}\label{eq-T_equilibre} 
T_{tr}=T_{eq}+O(\Kn),\ T_{rot}=T_{eq}+O(\Kn),\ T_{vib}=T_{eq}+O(\Kn),\ T_{tr,rot}=T_{eq}+O(\Kn).
\end{equation}

The Navier-Stokes equations are obtained by looking for a second order
expansion of $f$. In the following section, we first derive useful
second order expansions of energies and tensor $\Pi$ that are used in
our model.

\subsection{Energy and tensor relations at second order}
First,~\eqref{eq: cinetique} is multiplied by $\frac12|v-u|^2$,
$\varepsilon$, and $iRT_0$ and integrated w.r.t $v$, $\varepsilon$, and
$i$. We use relations~(\ref{eq: fequilibre2}),~(\ref{eq-T_equilibre}), and~\eqref{conserv} to get
\begin{equation}  \label{eq-dtEalpha1} 
\partial_t\begin{pmatrix} E_{tr}(f)\\ E_{rot}(f)\\ E_{vib}(f)\end{pmatrix}+u\cdot\nabla\begin{pmatrix} E_{tr}(f)\\ E_{rot}(f)\\ E_{vib}(f)\end{pmatrix}+\begin{pmatrix} RT_{eq}\nabla\cdot u\\ 0\\ 0\end{pmatrix}+O(\Kn)
=\frac{\Gamma}{\Kn\tau}
\begin{pmatrix} e_{tr}(T_{eq})\\ e_{rot}(T_{eq})\\ e_{vib}(T_{eq})\end{pmatrix}+\frac{1}{\Kn\tau} D
\begin{pmatrix} E_{tr}(f)\\ E_{rot}(f)\\ E_{vib}(f)\end{pmatrix}
\end{equation}
with
\begin{equation*}
D=\begin{pmatrix}
  -\frac{(1-\Gamma)\delta\theta}{3+\delta}-\Gamma &
  \frac{3(1-\Gamma)\theta}{3+\delta} & 0\\
  \frac{\delta(1-\Gamma)\theta}{3+\delta} &
  -\frac{3(1-\Gamma)\theta}{3+\delta}-\Gamma & 0 \\ 0&
  0&-\Gamma\end{pmatrix}.
  \end{equation*}
Note
that the eigenvalues of $D$ are $-\Gamma$, $-\Gamma$, and
$-\Gamma-(1-\Gamma)\theta$ so that~(\ref{eq-dtEalpha1}) is indeed a
relaxation process, and also that $D$ is invertible.

Moreover, from~\eqref{eq-T_tr_rot_vib}, we deduce the differential
relation $dE_{\alpha}(f) = e'_{\alpha}(T_{\alpha})dT_{\alpha}$, for $\alpha=tr,
rot, vib$. Then, using~\eqref{eq-T_equilibre} and the last equation
of~\eqref{eq-euler_nc}, we get
\begin{equation} \label{eq-dtEalpha2} 
\partial_t\begin{pmatrix} E_{tr}(f)\\ E_{rot}(f)\\ E_{vib}(f)\end{pmatrix}+u\cdot\nabla\begin{pmatrix} E_{tr}(f)\\ E_{rot}(f)\\ E_{vib}(f)\end{pmatrix}
=-\begin{pmatrix} e_{tr}'(T_{eq})CT_{eq}\\ e_{rot}'(T_{eq})CT_{eq}\\ e_{vib}'(T_{eq})CT_{eq} \end{pmatrix}\nabla\cdot u+O(\Kn).
\end{equation}

Finally, relations~(\ref{eq-dtEalpha1}) and~(\ref{eq-dtEalpha2}) give
the following system
\begin{eqnarray*} 
\Gamma
\begin{pmatrix} e_{tr}(T_{eq})\\ e_{rot}(T_{eq})\\ e_{vib}(T_{eq})\end{pmatrix}+ D
\begin{pmatrix} E_{tr}(f)\\ E_{rot}(f)\\ E_{vib}(f)\end{pmatrix}
=-\Kn\tau\begin{pmatrix} e_{tr}'(T_{eq})CT_{eq}-RT_{eq}\\ 
e_{rot}'(T_{eq})CT_{eq}\\ e_{vib}'(T_{eq})CT_{eq} \end{pmatrix}
\nabla\cdot u+O(\Kn^2)
\end{eqnarray*}
that has to be solved to get second order expansion of energies as
functions of the equilibrium temperature and of the divergence of $u$. We
only write here the relations that will be useful to derive the
Navier-Stokes hydrodynamics:
\begin{align*}
& E_{tr}(f)=e_{tr}(T_{eq})+\frac{\Kn\tau}{\Gamma}
          \left(\frac{3}2C-\frac{1}{\Gamma+(1-\Gamma)\theta}\left(\Gamma+\frac{3(1-\Gamma)\theta}{3+\delta}\right)\right)
                 RT_{eq}\nabla\cdot u +O(\Kn^2), \\
& E_{tr,rot}(f)=e_{tr,rot}(T_{eq})+\frac{\Kn\tau}{\Gamma}\left(\frac{3+\delta}2C-1\right) RT_{eq}\nabla\cdot u +O(\Kn^2).
\end{align*}
Similar relations are readily derived for temperatures $T_{tr}$ and
$T_{tr,rot}$ by using~\eqref{eq-T_tr_rot_vib} and~\eqref{eq-T_trrot},
and therefore,~\eqref{eq: Pi} can now be used to derive the second order
expansion of tensor $\Pi$:
\begin{equation} 
  \begin{split}
\Pi=&(1-(1-\Gamma)(1-\theta)\nu)RT_{eq}I+(1-\Gamma)(1-\theta)\nu
\Theta \\
&+(1-\Gamma)\theta \frac{\Kn\tau}{\Gamma}\left(C-\frac2{3+\delta}\right) RT_{eq}\nabla\cdot u\\
&+(1-\Gamma)(1-\theta)(1-\nu)\frac{\Kn\tau}{\Gamma}\left(C-\frac{1}{\Gamma+(1-\Gamma)\theta)}\frac2{3}\left(\Gamma+\frac{3(1-\Gamma)\theta}{3+\delta}\right)\right) RT_{eq}\nabla\cdot u\\
&+O(\Kn^2).
\end{split}
\end{equation}

Finally, we find it convenient to define the following three quantities
\begin{equation}\label{eq-gammas} 
 \gamma_{mono}=\frac{5}{3},\quad \gamma_{rot}=\frac{5+\delta}{3+\delta},\quad \gamma=1+\frac{R}{c_v(T_{eq})}\\
\end{equation}
that are nothing but heat capacity ratios for a monoatomic gas,
a polyatomic gas with rotational modes only, and the present gas with
rotational and vibrational modes, respectively. Then $\Pi$ can be
rewritten as
\begin{equation} \label{eqnpi}
\begin{split}
\Pi= & (1-(1-\Gamma)(1-\theta)\nu)RT_{eq}I+(1-\Gamma)(1-\theta)\nu
\Theta \\
&-\left(\frac{(1-\Gamma)(1-\theta)(1-\nu)}{1-(1-\Gamma)(1-\theta)}(\gamma_{mono}-\gamma_{rot})+\frac{(1-\Gamma)(1-(1-\theta)\nu)}{\Gamma}(\gamma_{rot}-\gamma)\right)\Kn\tau RT_{eq}\nabla\cdot u\\
&+O(\Kn^2).
\end{split}
\end{equation}


\subsection{\label{sec: NS}Navier-Stokes limit}

We first state our main result. 
\begin{proposition} The moments of $f$, solution of the ES-BGK
  model~\eqref{eq: cinetique}, satisfy the compressible Navier-Stokes
  equations up to $O(\Kn^2)$:
\begin{align*}
&\partial_t \rho + \nabla \cdot (\rho u)=O(\Kn^2), \\
&\partial_t (\rho u)+\nabla \cdot (\rho u\otimes u)+\nabla p=-\nabla \cdot\sigma+O(\Kn^2),\\
&\partial_t \mathcal{E}+\nabla \cdot (\mathcal{E}+p)u=-\nabla \cdot q-\nabla \cdot(\sigma u)+O(\Kn^2),
\end{align*}
where, in dimensional form, the shear stress tensor and the heat flux are given by
\begin{equation*}
\sigma= -\mu\left(\nabla u+(\nabla u)^T-\alpha \nabla \cdot u I\right),\quad  q=-\kappa\nabla T,
\end{equation*}
the viscosity and heat transfer coefficient are
\begin{equation*}
\mu=\frac{\tau p}{1-(1-\Gamma)(1-\theta)\nu},\qquad\kappa=(1-(1-\Gamma)(1-\theta)\nu) \mu c_p,
\end{equation*}
the second viscosity coefficient is
\begin{equation*}
\alpha=(\gamma-1)-\frac{(1-\Gamma)(1-\theta)(1-\nu)}{1-(1-\Gamma)(1-\theta)}(\gamma_{mono}-\gamma_{rot})-\frac{(1-\Gamma)(1-(1-\theta)\nu)}{\Gamma}(\gamma_{rot}-\gamma),
\end{equation*}
and the Prandtl number is
\begin{equation*}
Pr=\frac{\mu c_p}{\kappa}=\frac{1}{1-(1-\Gamma)(1-\theta)\nu},
\end{equation*}
while $c_p=\frac{\partial h}{\partial T_{eq}} $ is the heat capacity
at constant pressure, where $h=e(T_{eq}) + p/\rho$ is the
enthalpy. The heat capacity ratios $\gamma$, $\gamma_{mono}$,
$\gamma_{rot}$ are defined in~\eqref{eq-gammas}.
\end{proposition}

\begin{proof}
First,~(\ref{eq: nd_kinetic}) yields
\begin{equation*}
f=\mathcal{G}[f]-\tau \Kn(\partial_t \mathcal{M}[f]+v\cdot \nabla \mathcal{M}[f])+O(\Kn^2)
\end{equation*}
By linearity, the stress tensor and the heat flux are
\begin{equation}
\begin{aligned}
\Sigma(f)&=\Sigma(\mathcal{G}[f])-\tau \Kn\,\Sigma(\partial_t \mathcal{M}[f]+v\cdot \nabla \mathcal{M}[f])+O(\Kn^2)\\
q(f)&=q(\mathcal{G}[f])-\tau \Kn\,q(\partial_t \mathcal{M}[f]+v\cdot \nabla \mathcal{M}[f])+O(\Kn^2)
\end{aligned}
\end{equation}

We first deal with the expansion of the stress tensor. For the first
term, note that (\ref{eq-Gaussalpha}) and~\eqref{eq: Pi} imply
$\Sigma(\mathcal{G}[f]) = \rho \Pi$. Therefore the expression
above reads
\begin{equation*}
\Sigma(f)=\rho \Pi-\tau \Kn \la (v-u) \otimes (v-u) (\partial_t \mathcal{M}[f]+v\cdot \nabla \mathcal{M}[f]) \ra_{v,\varepsilon,i}+O(\Kn^2).  
  \end{equation*}
For the second term, tedious but
standard calculations show that time derivatives can be written as
functions of the space derivatives only by using Euler
equations~(\ref{eq-euler}), and then suitable integral formula give
\begin{equation*}
\Sigma(f)= \rho\Pi-\tau \Kn \rho R T_{eq} (\nabla u + (\nabla u)^T - C\nabla \cdot u I) +O(\Kn^2),
\end{equation*}
see some details in appendix~\ref{app:Sigmaq} and~\ref{app:integrals}. Then combining this
equation with~(\ref{eqnpi}) one finally gets
\begin{equation*}
  \Sigma(f)= \rho R T_{eq} I
  -\Kn\tau \rho R T_{eq} \frac{1}{1-(1-\Gamma)(1-\theta)\nu} (\nabla u+(\nabla u)^T -\alpha\nabla\cdot u I) +O(\Kn^2),
\end{equation*}
where $\alpha$ takes the value given in the proposition.
Now we use the equilibrium pressure $p=\rho R T_{eq}$ and we define the viscosity
coefficient $\mu = \tau p/(1-(1-\Gamma)(1-\theta)\nu)$ to get the
value of the shear stress tensor given in the proposition.

For the heat flux, a simple parity argument shows that
$q(\mathcal{G}[f]) = 0$, so that
\begin{equation*}
q(f) = -\tau \Kn\la (\frac{1}{2}|v-u|^2+\varepsilon_r+iRT_0 )(v-u)
(\partial_t \mathcal{M}[f]+v\cdot \nabla
\mathcal{M}[f])\ra_{v,\varepsilon,i}+O(\Kn^2).
\end{equation*}
Using the same tools as for the stress tensor, we find
\begin{equation*}
q(f)= -\tau\Kn p\nabla\left(\frac{5+\delta+\delta_v(T_{eq})}{2}RT_{eq}\right) +O(\Kn^2).  
\end{equation*}
Now we notice that $\frac{5+\delta+\delta_v(T_{eq})}{2}RT_{eq} =
e(T_{eq}) + RT_{eq} = e(T_{eq}) + p/\rho = h(T_{eq})$. Consequently,
\begin{equation*}
  \begin{split}
  q(f)& = -\tau\Kn  p\nabla h(T_{eq}) +O(\Kn^2)   \\
  & =   -\tau\Kn  p \frac{\partial h}{T_{eq}} \nabla T_{eq}  +O(\Kn^2) = -\Kn
  \tau p c_p  \nabla T_{eq}  +O(\Kn^2) ,
\end{split}
\end{equation*}
which gives the Fourier law with the value of the heat transfer
coefficient $\kappa = \tau p c_p$ in dimensional variables.  Then
using the value of $\mu$ found above leads to the value of $\kappa$ given in the proposition.

Finally, note that with this analysis, if the Prandtl number is defined as
$\Pr  = \mu c_p/\kappa$, then we find $\Pr =
\frac{1}{1-(1-\Gamma)(1-\theta)\nu}$, which is the same result as found
in section~\ref{sec:relaxf}.

\end{proof}

\begin{remark}
  Note that by writing $\nu$, $\Gamma$ and $\Theta$ as functions of the
Prandtl number and of  $Z_{rot}$ and $Z_{vib}$ (see
section~\ref{subsec:relax}), the second viscosity can be simply written
\begin{equation*}
\alpha = \frac23-\frac{Z_{rot}}{Pr}(\gamma_{mono}-\gamma_{rot})-\frac{Z_{vib}}{Pr}(\gamma_{rot}-\gamma)
\end{equation*}
This second viscosity appears to be driven by relaxation processes due to rotations and vibrations of molecules characterized by $Z_{rot}$ and $Z_{vib}$. 
\end{remark}

\section{\label{sec: r_esbgk}Reduced ES-BGK model}

\subsection{The reduced distribution technique}

For numerical simulations with a deterministic solver, our ES-BGK
model is much too expensive, since the
distribution $f$  depends on many variables: time $t\in\mathbb{R}$, position
$x\in\mathbb{R}^3$, velocity $v\in\mathbb{R}^3$, rotational energy
$\varepsilon\in\mathbb{R}^+$ and discrete levels of the vibrational
energy $i\in\mathbb{N}$. For aerodynamic problems, it is generally 
sufficient to compute the macroscopic velocity and temperatures fields : a
reduced distribution technique~\cite{HH_1968} (by integration w.r.t rotational
and vibrational energy) permits to drastically reduce the
computational cost, without any approximation. We define the three
marginal distributions 
\begin{equation*}
  \begin{pmatrix}
F(t,x,v) \\ G(t,x,v)\\ H(t,x,v)
 \end{pmatrix}
 = \sum_{i=0}^{+\infty}\int_{\mathbb{R}}
 \begin{pmatrix}
1 \\ \varepsilon \\i R T_0
 \end{pmatrix}
  f(t,x,v,\varepsilon,i) \, d\varepsilon.
\end{equation*}
The macroscopic quantities defined
by~\eqref{eq-mtsf}--\eqref{eq-Theta_q} can now be computed
through $F$, $G$ and $H$ only by
\begin{equation}
\begin{aligned}
  & \rho=\la F \ra_{v},\quad \rho u=\la vF \ra_{v},\\
  & \rho E_{tr}(f)=\la\frac{1}{2}|v|^2 F \ra_{v},
  \quad \rho E_{rot}(f)= \la G \ra_{v},\quad \rho E_{vib}(f)=\la H \ra_{v}, \\
& \rho \Theta =\la\frac{1}{2}(v-u)\otimes (v-u)F \ra_{v},\quad q= \la (\frac{1}{2}|v-u|^2 F+G+H)(v-u) \ra_{v},
\end{aligned}
\end{equation}
where $\la . \ra_{v}$ denotes integrals with respect to $v$ only.

The reduced ES-BGK is obtained by multiplying our kinetic model~\eqref{eq: cinetique}-\eqref{eq: operator}
by the vector $(1,\varepsilon,iRT_0)^T$ and by summing and integrating
w.r.t to $i$ and $\varepsilon$, respectively. We get:
\begin{equation}
\label{eq: r_model}
\partial_t \F +v\cdot \nabla \F = \frac{1}{\tau}(\mathcal{G}[\F]-\F),
\end{equation}
where $\F = (F,G,H)$ and $ \mathcal{G}[\F] = (\mathcal{G}_{tr}[f],e_{rot}^{rel}\mathcal{G}_{tr}[f],e_{vib}^{rel}\mathcal{G}_{tr}[f])$.

\subsection{Reduced entropy}

In this section, we again use the change of variable
$\varepsilon=I^{2/\delta}$. To prove the H-theorem for our reduced model, it is convenient to
view it as an entropic moment closure (w.r.t variables $I$ and $i$),
see for instance~\cite{Levermore,perthame_1,mieussens99}. Then we
define $g_{\F}$ such that $\ent(g_{\F})$ is the minimum of
$\ent$ on the set $\chi_{\F} = \lbrace \phi\geq 0 \text{ such that}  \la
(1,I^{\frac{2}{\delta}},iRT_0)\phi\ra_{I,i} = \F  \rbrace$, and we set
$\entred(\F) = \ent(g_{\F})$. It is now possible to prove that
$\entred(\F)$ is an entropy for our reduced system.

\begin{proposition}[Reduced entropy]
An explicit form of $\entred$ is given by $\entred(\F) = \la h(\F) \ra_v$, where
$h$ is the strictly convex function defined by
\begin{equation}\label{eq-hFGH} 
\begin{aligned}  
h(\F) =& F\left[
  \left(1+\frac{\delta}{2}\right)\left(\log\left(\frac{F}{G^{\frac{\delta}{2+\delta}}}\right)-1\right)+\log\left(\frac{RT_0F}{RT_0F+H}\right)+\frac{\delta}{2}\log\frac{\delta}{2}+\log
  \Lambda(\delta) \right] \\
& +\frac{H}{RT_0}\log\left(\frac{H}{RT_0F+H}\right).
\end{aligned}
\end{equation}
\end{proposition}

\begin{proof}
  First, we compute $g_{\F}$ by solving the minimization problem
  $\ent(g_{\F}) = \min_{\chi_{\F}} \ent$. Since $\chi_{\F}$ is convex, we use a Lagrange multiplier method to find the minimum of the functional $\mathcal{L}$ defined as follows:
\begin{equation*}
\mathcal{L}(\phi,\alpha,\beta,\gamma)=\la \phi\log \phi-\phi \ra_{I,i} +\alpha \left(\la \phi \ra_{I,i}-F\right)+\beta \left(\la I^{2/\delta} \phi \ra_{I,i}-G\right)+\gamma \left(\la iRT_0 \phi \ra_{I,i}-H\right),
\end{equation*}
where the Lagrange multipliers $\alpha$, $\beta$ and $\gamma$ are functions of $v$. The minimum satisfies $\displaystyle \frac{\partial
  \mathcal{L}}{\partial \phi}(g_{\F},\alpha,\beta,\gamma)=0$, which leads to
\begin{equation}
\label{eq: r_g}
g_{\F}=\exp(-\alpha-\beta I^{2/\delta}-\gamma iRT_0 ).
\end{equation}
With the linear constraints $\la(1,I^{2/\delta},RT_0i)g_{\F} \ra_{I,i} =
(F,G,H)$, we find explicit values for $\alpha$, $\beta$, and $\gamma$
as functions of
$F$, $G$, and $H$.
Consequently, by using $\entred(\F) = \ent(g_{\F})$ and these
values of $\alpha$, $\beta$, and $\gamma$, we find~(\ref{eq-hFGH}).
\end{proof}

\begin{remark}
The convexity property of $h$ could also be proved  without any
explicit computation: indeed, it can be viewed as the Legendre transform of
$h^*(\alpha, \beta, \gamma) = \la \exp(-\alpha - \beta I^{2/\delta}- \gamma iRT_0)\ra_v$
(where $\alpha$, $\beta$, and $\gamma$ are such that $\F = \la(1,I^{2/\delta},iRT_0)
\exp(-\alpha - \beta I^{2/\delta}- \gamma iRT_0)\ra_{I,i}$), which is clearly strictly convex (see details for a similar argument in~\cite{Levermore}).
\end{remark}

\begin{proposition}[H-theorem]
  The reduced ES-BGK system~\eqref{eq: r_model} satisfies the following local entropy
  dissipation law
  \begin{equation}\label{eq-Htheo-reduced} 
    \partial_t \entred(\F) + \nabla \cdot \la v h(\F)\ra_v =
    \la D_{\F}h(\F)(    \frac{1}{\tau}\mathcal{G}[\F] - \F)
    \ra_v
    \leq 0,
  \end{equation}
and the equilibrium is
reached (the right-hand side of~\eqref{eq: r_model} is zero) if, and
only if,
\begin{equation*}
\F = (\mathcal{M}_{tr}[f],
e_{rot}(T_{eq})\mathcal{M}_{tr}[f], e_{vib}(T_{eq})\mathcal{M}_{tr}[f]),  
  \end{equation*}
where $\mathcal{M}_{\mo}[f]$ is the Maxwellian for translation d.o.f
(see section~\ref{subsec:construct}).
  \end{proposition}

  \begin{proof}
The equality in~\eqref{eq-Htheo-reduced} is obtained with elementary
calculus. Since $h$ is convex, the right-hand side of this equality
satisfies
\begin{equation*}
  \la D_{\F}h(\F)( \frac{1}{\tau}\mathcal{G}[\F] - \F)\ra_v
  \leq  \la  h(\mathcal{G}[\F])- h(\F)\ra_v 
     =  \entred(\mathcal{G}[\F])- \entred(\F)
  \end{equation*}
Therefore, the H-theorem is proved if we can prove that this entropy
difference is non-negative.

First, we prove that $ \entred(\mathcal{G}[\F]) \leq 
    \ent({\cal G}[g])$. Indeed, ${\cal G}[g]$ is clearly in
    $\chi_{\mathcal{G}[\F]}$,
      and since $\entred(\mathcal{G}[\F])$ is the
    minimum value of $\ent$ on this set, we have $\entred(\mathcal{G}[\F]) \leq
    \ent({\cal G}[g])$. It is easy to prove that we have in fact
    equality, but this is not necessary here.

Now it is sufficient to prove that $\ent(({\cal G}[g])) \leq
\entred(\F)$. First, remind that in the proof of
Proposition~\ref{prop:theoH} (step 2), we have obtained
\begin{equation*}
    \ent(({\cal G}[g]))   =
    S(\rho,u,\Pi,T_{rot}^{rel},T_{vib}^{rel})
    \leq S(\rho,u,\Theta,T_{rot},T_{vib}).
\end{equation*}
Then we remind that $\entred(\F) =
\ent(g_{\F})$, where $g_{\F}$ is in $\chi_{\F}$. Consequently,
$g_{\F}$ has the same moments as $g$, and hence $S(\rho,u,\Theta,T_{rot},T_{vib})\leq
\ent(g_{\F}) = \entred(\F)$, which concludes the proof. 
\end{proof}

  \begin{remark}
The reduced entropy can be simplified by dropping out some terms that
are proportional to $F$: if we set
\begin{equation*}
  \begin{aligned}
\tilde{\entred}(\F)=F\log\left(\frac{F}{G^{\frac{\delta}{2+\delta}}}\right)-F+F\log\left(\frac{RT_0F}{RT_0F+H}\right)+\frac{H}{RT_0}\log\left(\frac{H}{RT_0F+H}\right),
\end{aligned}
\end{equation*}
then $\tilde{\entred}$ is also strictly convex. The previous
proof also leads to an entropy production term lower than $\tilde{\entred}(\mathcal{G}[\F])-
    \tilde{\entred}(\F).$ This entropy difference is the same as
    that obtained with the original reduced entropy $\entred$ up to
    an integral of $\mathcal{G}[\F] - \F$ which is zero (mass
    conservation). This simplified reduced entropy is similar to that
    of~\cite{esbgk_poly,mieussens99} with, in addition, the effects of vibrations (\cite{Mathiaud2019}).
\end{remark}

\section{\label{sec: numeric}Numerical test}

 In this section, we study the relaxation process to equilibrium in a
 space homogeneous polyatomic vibrating gas by using Monte Carlo
 simulations of the ES-BGK model presented in section \ref{sec:
   esbgk}. Our results will be used to confirm that the relaxation
 rates of translational, rotational, and vibrational degrees of
 freedom can indeed be obtained by adjusting the parameters $\theta$
 and $\Gamma$. Moreover, we will also check that the correct Prandtl number
 can be obtained by adjusting the parameter $\nu$.

In this space homogeneous case, the ES-BGK model reads 
\begin{equation}
\label{eq: ec_homogene}
\partial_t f=\frac{1}{\tau}(\mathcal{G}[f]-f).
\end{equation}
Note that by conservation property~\ref{prop:cons}, the mass density,
velocity, and equilibrium temperature, are constant in time here.

\subsection{The Monte Carlo method}
To observe the process of relaxation we enforce a non-equilibrium
initial condition, for instance a gap between the mean of the
velocities of the particles and the velocity of the gas: the model
should relax velocities and internal energies towards equilibrium
state. We use a large number $N$ of numerical particles related to
the real molecules by a distribution function associated to a
constant numerical weight $\omega=1/N$. We use an explicit Euler
scheme for time discretization and get:
\begin{equation}
\label{eq: ec_homogene2}
f^{n+1}=\left(1-\frac{\Delta t}{\tau}\right)f^n+\frac{\Delta t}{\tau}\mathcal{G}[f^n],
\end{equation}
with $\Delta t=t^{n+1}-t^n$ and we consider $\Delta t/\tau \leq 0.1$
to ensure stability~\cite{LBP}.  Equation~(\ref{eq: ec_homogene2})
models the effects of collisions on the distribution functions of
velocities and energies: at time $t^{n+1}$ the distribution function
is a convex combination of the distribution function at time $t^n$
and its corresponding local Gaussian distribution. This can be simulated
with a Monte Carlo algorithm as follows: at each time step, for each
particle, we decide if its velocity has to be modified by a collision
(with a probability $\Delta t/\tau$). In such case, the components of
its velocity $v$ are modified by 
\begin{eqnarray}
v_k=u_k+A   (B_1,B_2,B_3)^T,
e_{rot}^k=B_4, e_{rot}^k=B_5,
\end{eqnarray}
where $u$ is macroscopic velocity of the gas, $B_1$, $B_2$, $B_3$ are
three random numbers generated from a standard normal law and the
matrix $A$ needs to satisfy the condition: $\Pi=AA^T$ (generally, $A$
is given by the Cholesky decomposition due to its simplicity and its
low computational cost). $B_4$ is generated through an exponential distribution depending on $RT_{rot}^{rel}$  and $B_5$ through a Poisson distribution of parameter $RT_{vib}^{rel}$.

\subsection{Numerical results}

We consider $N=10^7$ numerical particles of velocities initially
distributed according to a Gaussian distribution of variance $500$ and
of mean $0$ for the second and the third components and $50$ for the
first. The initial rotational energy is set to $1000\,r_1$ and the
initial vibrational energy is set to $10\,r_2$ where the random
numbers $r_1$ and $r_2$ follow an uniform law between $0$ and $1$. The
parameters $\theta$ and $\Gamma$ are defined
by~\eqref{eq-defGammatheta}, so that collision numbers $Z_{rot}$ and
$Z_{vib}$ are respectively equal to $5$ and $20$. Finally, we set
$\nu$ according to~\eqref{eq-Pr} so that the Prandtl number is equal to $0.73$, which is close to the tabulated value
for air at $2000K$.
These non-equilibrium initial conditions create energy exchanges between
modes and a heat flux.

We first show in figure~\ref{fig: temp} that the temperature relaxes as
expected (see section~\ref{sec:relax}). First, the translational directional
temperatures converge to the mean translational temperature $T_{tr}$
at time $\tau$. Then, at time $20\tau$, this temperature and the rotational temperature converge towards the
translational-rotational temperature $T_{tr,rot}$. Finally, at time
$100 \tau$, $T_{tr,rot}$ and the vibrational temperature $T_{vib}$
converge to the equilibrium temperature $T_{eq}$.

In figure~\ref{fig: equil}, we show the distribution of velocities,
rotational energy, and vibrational energy, obtained at steady
state. This distributions are compared to  the
components of the Maxwellian distribution~\eqref{eq-Mf}, and we
observe a prefect agreement between them, which proves
that the correct equilibrium is captured by the model.


Now we plot in figure~\ref{fig:diff_temp} the temperature differences
$T_{tr}-T_{tr,rot}$, $T_{rot}-T_{tr,rot}$,  $T_{tr,rot}-T_{eq}$, and
$T_{vib}-T_{eq}$. We observe that this functions converge
exponentially, as expected (even if a numerical noise is observed for
$t>20\tau$ which corresponds to machine accuracy when the
translational and the rotational temperatures are
converged). Moreover, according to~\eqref{eq:diff_temp}, the slopes of
these convergence curves can be used to compute $Z_{rot}$ and $Z_{vib}$,
a posteriori.  We find $Z_{rot}=4.878$ and $Z_{vib}=19.61$, which is
very close to the expected values.

Finally, we plot in figure~\ref{fig: Pr} the evolution of the difference of
the first directional temperature $T_{11}$ and the mean translational
temperature $T_{tr}$, as well as the evolution of the first component
of the heat flux $q_1$. According to equation~\eqref{eq-ratioqT}, it is possible to estimate
the Prandtl number by evaluating the slopes of the of these
quantities: we find $0.71$, which is close to the input value $0.75$.

\section{Conclusion}

In this paper, we have proposed an extension of the original
polyatomic ES-BGK model to take into account discrete levels of
vibrational energy. For a gas flow in non-equilibrium, for instance
for a high enthalpy flow, we expect this model to capture the
shock position and the parietal heat flux with more accuracy. This model satisfies the
conservation properties and the H-theorem and allows to adjust correct
transport coefficients and relaxation rates. It has been illustrated by
numerical simulations for an homogeneous problem.  Finally, a reduced
model which also satisfies the conservation laws and the H-theorem has
been obtained: with this model, it should be possible to make
simulations at a computational cost which is of same order of
magnitude as for a monoatomic gas.

\appendix

\section{Inequality for $\det(\Theta)/det(\Pi)$}
\label{subsec:ineqthetapi}

Here we prove the result  for inequality (\ref{ineqthetapi}) which is:
$\frac{\det \Theta}{\det\Pi}\leq
\left(\frac{E_{tr}(f)}{e_{tr}^{rel}}\right)^3$. We establish the
result in a basis where $\Theta$ can be diagonalized and we note
$\Theta_1,\Theta_2,\Theta_3$ its eigenvalues. Note that $\Pi$ is
diagonal in the same basis. Then we have
\begin{equation*}
\frac{\det \Theta}{\det\Pi} =  \frac{\prod_{i=1}^3\Theta_i}{\prod_{i=1}^3(\Gamma RT_{eq}+(1-\Gamma)(\theta(RT_{tr,rot})+(1-\theta)(\nu \Theta_i+(1-\nu)RT_{tr})}.
  \end{equation*}
  The proof is based on convexity arguments. However, since
  parameter $\nu$ can be negative (we remind that $\nu$ lies in
  $[-\frac{1}{2},1]$), we first want to obtain an lower bound for
  ${\det\Pi}$ that does not depend on $\nu$.

First, we consider $\det \Pi$ as a function of $\nu$,
and we take its logarithm denoted by $\phi(\nu)$: 
\begin{equation*}
  \phi(\nu) =  \sum_{i=1}^3 \log(\Gamma RT_{eq}+(1-\Gamma)(\theta(RT_{tr,rot})+(1-\theta)(\nu \Theta_i+(1-\nu)RT_{tr}).
\end{equation*}
By computing their second derivatives, it can easily be seen that each component of this sum is a concave
function of $\nu$, and so is the function $\phi$. Moreover, a simple
derivation and relation $\sum_{i=1}^3 \Theta_i=3RT_{tr}$
(see section~\ref{subsec:int_temp}) show that $\phi'(0)=0$. These two
properties imply that $\phi$
necessarily reaches its minimum on $[-\frac{1}{2},1]$ at $\nu=-\frac{1}{2}$ or at $\nu=1$.

Now we have to determine what is the minimum between $\phi(-\frac{1}{2})$ and
$\phi(1)$.  In order to simplify the notations, we introduce $X =
\Gamma RT_{eq}+(1-\Gamma)\theta RT_{tr,rot}$, which is positive, and
$Y = (1-\Gamma)(1-\theta)$, which is in $[0,1[$. Then we find
\begin{equation*}
  \phi(-\frac{1}{2}) = \log( \prod_{i=1}^3 (X + Y  \frac{\Theta_j+\Theta_k}{2}) )
\quad \text{ and } \quad 
\phi(1) = \log( \prod_{i=1}^3 (X + Y \Theta_i)),
\end{equation*}
where $j$ and $k$ in the first expression denote the two other indices
different from $i$. A convex inequality (which is nothing but the usual inequality between arithmetic and
geometric means) implies 
\begin{equation*}
    \phi(-\frac{1}{2}) \geq \log (\prod_{i=1}^3 (\sqrt{(X + Y
      \Theta_j} \sqrt{(X + Y  \Theta_k})) = \log( \prod_{i=1}^3 (X + Y
    \Theta_i)) = \phi(1).
\end{equation*}
Consequently, $\phi(\nu)\geq \phi(1)$ for every $\nu$ in
$[-\frac{1}{2},1]$: this implies $\det\Pi \geq \prod_{i=1}^3 (X + Y
    \Theta_i)$ and we deduce this upper bound
\begin{equation}\label{eq-upperbound} 
  \frac{\det \Theta}{\det\Pi} \leq \prod_{i=1}^3 \frac{\Theta_i}{X + Y
    \Theta_i},
\end{equation}
that does not depend on $\nu$ anymore, as announced above.

In the last part, we analyze the logarithm of the right-hand side of
the previous inequality: we denote by 
\begin{equation*}
  g(\Theta) = \log \prod_{i=1}^3 \frac{\Theta_i}{X + Y \Theta_i} =
  \sum_{i=1}^3 f(\Theta_i),
\end{equation*}
where $ f(s) = \log \left(   \frac{s}{X + Y s} \right)$ is clearly a concave function. Then we use
the Jensen inequality to get
\begin{equation*}
\begin{split}
   \frac{1}{3}g(\Theta) & =  \frac{1}{3} \sum_{i=1}^3 f(\Theta_i)  \leq f \left( \frac{1}{3} \sum_{i=1}^3 \Theta_i  \right)  \\
& = f(RT_{tr}) =  \log \left(   \frac{RT_{tr}}{X + Y RT_{tr}} \right).
\end{split}
\end{equation*}
Now we note that $X + Y RT_{tr} = RT_{tr}^{rel}$
(see the definition of $X$ and $Y$ above and the definition~(\ref{eq-Ttrrel})) of $T_{tr}^{rel}$, so that $g(\Theta)\leq \log (
(\frac{RT_{tr}}{RT_{tr}^{rel}} )^3)$. Finally, we use this estimate in~(\ref{eq-upperbound}) to find 
\begin{equation*}
  \frac{\det \Theta}{\det\Pi} \leq \left(\frac{RT_{tr}}{RT_{tr}^{rel}} \right)^3,
\end{equation*}
and this gives the result, since we remind that $E_{tr}(f) =
\frac{3}{2}RT_{tr}$ and $e_{tr}^{rel} = \frac{3}{2}RT_{tr}^{rel}$.

\section{First and second order partial derivatives of $\mathcal{S}$}
\label{sec:appS}
We remind that
\begin{equation*}
\mathcal{S}(e_1,e_2,e_3)=R\left(\frac32\log(e_1)+\frac{\delta}2\log(e_2)+\log\left(1+\frac{e_3}{RT_0}\right)+\frac{e_3}{RT_0}\log\left(1+\frac{RT_0}{e_3}\right)\right).
\end{equation*}
The first order derivatives of $S$ are
\begin{equation*}
  \partial_{1}\mathcal{S} = \frac{3}{2}R\frac{1}{e_1}, \qquad 
  \partial_{2}\mathcal{S} = \frac{\delta}{2}R\frac{1}{e_2}, \qquad 
  \partial_{3}\mathcal{S} = \frac{1}{T_0}\log\left(1+ \frac{RT_0}{e_3}\right).
\end{equation*}
At $(e_1,e_2,e_3) = (e_{tr}^{rel},e_{rot}^{rel},e_{vib}^{rel})$, with
the corresponding definitions~\eqref{eq-fermeture2} and~\eqref{eq-Ttrrel} of the relaxation
temperatures, the relations above give
\begin{equation*}
  \partial_{1}\mathcal{S} = \frac{1}{T_{tr}^{rel}}, \qquad 
  \partial_{2}\mathcal{S} = \frac{1}{T_{rot}^{rel}}, \qquad 
  \partial_{3}\mathcal{S} = \frac{1}{T_{vib}^{rel}},
\end{equation*}
while the second order derivatives are
\begin{equation*}
\partial_{1,1}\mathcal{S}=-\frac{3R}{2(e_{tr}^{rel})2},\qquad
\partial_{2,2}\mathcal{S}=-\frac{\delta R}{2(e_{rot}^{rel})^2},\qquad
\partial_{3,3}\mathcal{S}=-\frac{R}{e_{vib}^{rel}\left(RT_0+{e_{vib}^{rel}}\right)},
\end{equation*}
and are clearly negative, while the cross derivatives are zero.

\section{Second order expansion of $\Sigma(f)$ and $q(f)$}
\label{app:Sigmaq}

Since $\mathcal{M}[f]= \mathcal{M}_{\mo}[f]\mathcal{M}_{rot}[f]\mathcal{M}_{vib}[f]$,
the expansion of $\partial_t \mathcal{M}[f]+v\cdot \nabla
\mathcal{M}[f]$ requires the expansion of the transport operator
applied to each component of $\mathcal{M}[f]$. We only detail here how we
proceed for the translation component $\mathcal{M}_{\mo}[f]$. The chain rule gives
\begin{equation*}
  \begin{split}
    \partial_t \mathcal{M}_{\mo}[f]+v\cdot \nabla \mathcal{M}_{\mo}[f]
 = & \left[\frac{\partial_t \rho + v\cdot\nabla
     \rho}{\rho}
 + (\partial_t u + (v\cdot\nabla) u ) \cdot \frac{v-u}{RT_{eq}}\right.\\
& \quad \left. + (\partial_t T_{eq} + v\cdot\nabla T_{eq} ) \left(
    \frac{|v-u|^2}{2RT_{eq}}-\frac{3}{2}\right)\frac{1}{T_{eq}}\right]
\mathcal{M}_{\mo}[f].
  \end{split}
  \end{equation*}
Euler equations~\eqref{eq-euler_nc} are used to replace time
derivatives of $\rho$, $u$, and $T_{eq}$ by their space derivatives,
and finally, we use the change of variables
$V=\frac{v-u}{\sqrt{RT_{eq}}}$ to get
\begin{equation*}
\partial_t M_{\mo}[f]+v\cdot \nabla M_{\mo}[f]=\frac{\rho}{(RT_{eq})^{3/2}}M_0(V)\left( A(V)\cdot\frac{\nabla \theta}{\sqrt{\theta}}+B(V):\nabla u  \right)+O(\Kn),
\end{equation*}
with
\begin{equation*}
A(V)=\left( \frac{|V|^2}{2}-\frac{5}{2} \right)V, \qquad \text{ and }\qquad
B(V)=V \otimes V-\left( \left( \frac{|V|^2}{2}-\frac{3}{2} \right)C+1 \right)I  .
\end{equation*}

The same kind of algebra is also used for the components
$\mathcal{M}_{rot}[f]$ and $\mathcal{M}_{vib}[f]$. They are much simpler
and are left to
the reader.

\section{Gaussian integrals and other summation formulas}
\label{app:integrals}
In this section, we give some summation and integrals formula that are
used in the paper. First, we have $\sum_{i=0}^{+\infty} e^{-i\theta} =
\frac{1}{1-e^{-\theta}}$ and $\sum_{i=0}^{+\infty} ie^{-i\theta} =
\frac{e^{-\theta}}{(1-e^{-\theta})^2}$, which can be used to obtain
\begin{equation*}
 \sum_{i=0}^{+\infty}   {\cal M}_{vib}(i) [f] = 1, 
  \qquad \text{ and }   \qquad  \sum_{i=0}^{+\infty}   iR T_0{\cal M}_{vib}(i) [f] = \frac{\delta_v(T_{eq})}{2}RT_{eq}.
  \end{equation*}

Then, we remind the gamma function ${\bf \Gamma}(x) =
\int_{0}^{+\infty}s^{x-1}e^{-s}\, ds$, which is such that ${\bf \Gamma}(x+1) = x{\bf \Gamma}(x)$ and ${\bf
  \Gamma}(1) = 1$. This is used to get
\begin{equation*}
 \int_0^{+\infty} {\cal M}_{rot}[f](\varepsilon) \, d\varepsilon = 1   \qquad \text{ and }   \qquad
 \int_0^{+\infty} \varepsilon {\cal M}_{rot}[f](\varepsilon) \, d\varepsilon = \frac{\delta}{2}RT_{eq}.
 \end{equation*}

Finally, we remind the definition of the absolute Maxwellian $M_0(V) =
\frac{1}{(2\pi)^{\frac32}}\exp(-\frac{|V|^2}{2})$. We denote by
$\cint{\phi}_V = \int_{\R^3}\phi(V)\, dV$ for any function $\phi$. It is standard to
derive the following integral relations (see~\cite{chapmancowling},
for instance), written with the Einstein notation:
\begin{align*}
&   \cint{M_0}_V = 1, \\
&   \cint{V_iV_jM_0}_V = \delta_{ij}, \qquad \cint{V_i^2M_0}_V = 1,
  \qquad \cint{|V|^2M_0}_V = 3, \\
& \cint{V_iV_jV_kV_lM_0}_V = \delta_{ij}\delta_{kl}  +
  \delta_{ik}\delta_{jl}  + \delta_{il}\delta_{jk} , \qquad \cint{V_i^2V_j^2M_0}_V = 1 + 2\, \delta_{ij} \\
& \cint{V_iV_j|V|^2M_0}_V = 5 \,\delta_{ij},  \qquad \cint{|V|^4M_0}_V
  = 15, \\
& \cint{V_iV_j|V|^4M_0}_V = 35 \,\delta_{ij},  \qquad \cint{|V|^6M_0} = 105,
\end{align*}
while all the integrals of odd power of $V$ are zero. Note that the first relation of each line implies the other relations of the same line: these relations are given here to improve the readability of the paper.
From the previous Gaussian integrals, it can be shown that for any
$3\times 3$ matrix $C$, we have
\begin{equation*}
\cint{V_iV_jC_{kl}V_kV_lM_0}_V = C_{ij} + C_{ji} + C_{ii}\delta_{ij}.
\end{equation*}


\bibliographystyle{unsrt}
\bibliography{biblio}

\newpage
\begin{figure}[htbp]
\centering 
\includegraphics[scale=0.6]{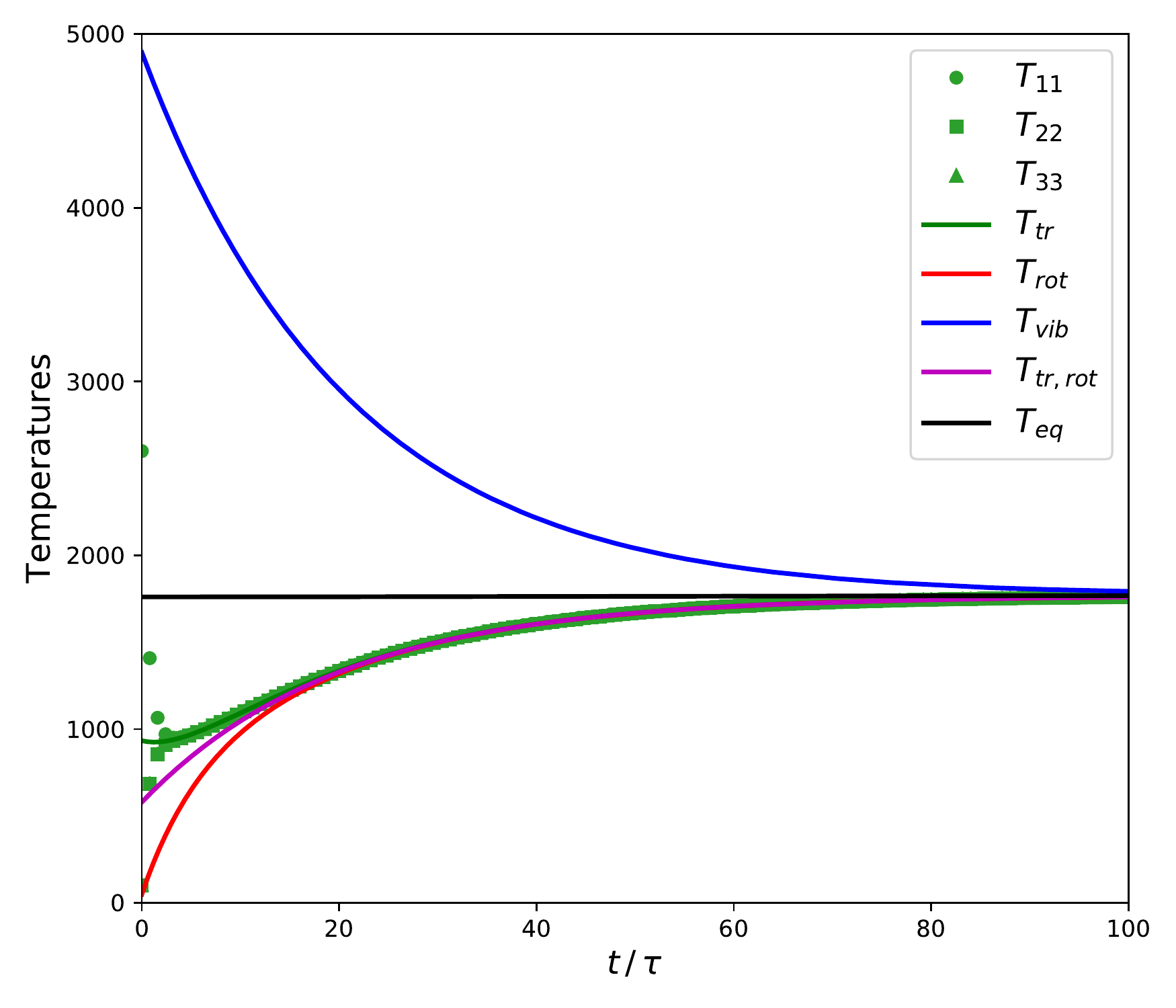}\hfill
\includegraphics[scale=0.6]{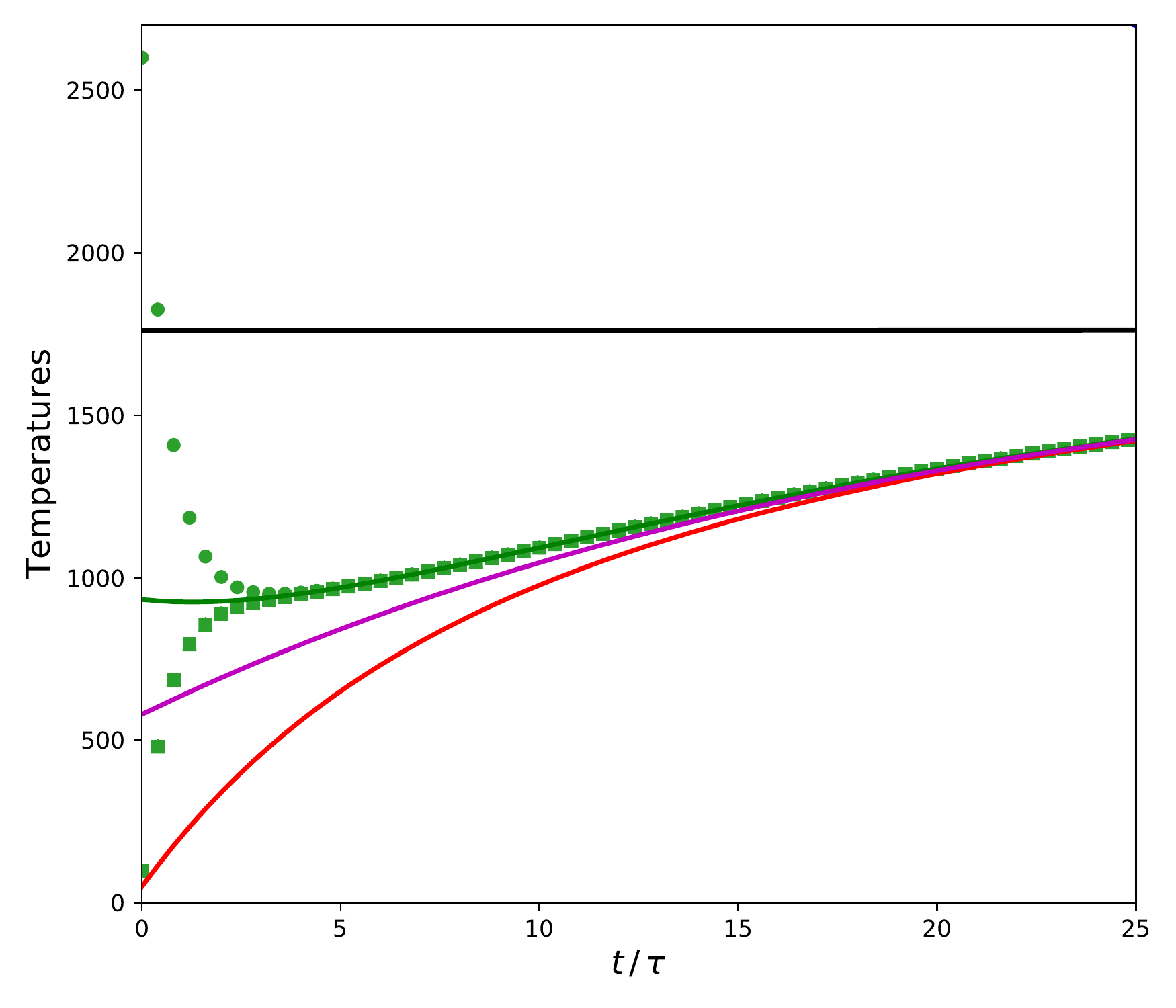}
\caption{\label{fig: temp}Relaxation of temperatures, on the right a
  zoom between $t=0$ and $t=25\tau$. ($\circ$) $T_{11}$, ($\square$)
  $T_{22}$ and ($\triangle$) $T_{33}$ are the components of the stress
  tensor. $T_{tr}$ (green), $T_{rot}$ (red)  and  $T_{vib}$ (blue) are
  respectively the temperatures of translation, rotation and
  vibration, while $T_{tr,rot}$ (purple) and  $T_{eq}$ (black) are
  the translational-rotational temperature and  temperature at
  equilibrium, respectively}
\end{figure}

\begin{figure}[htbp]
\centering 
\includegraphics[scale=0.45]{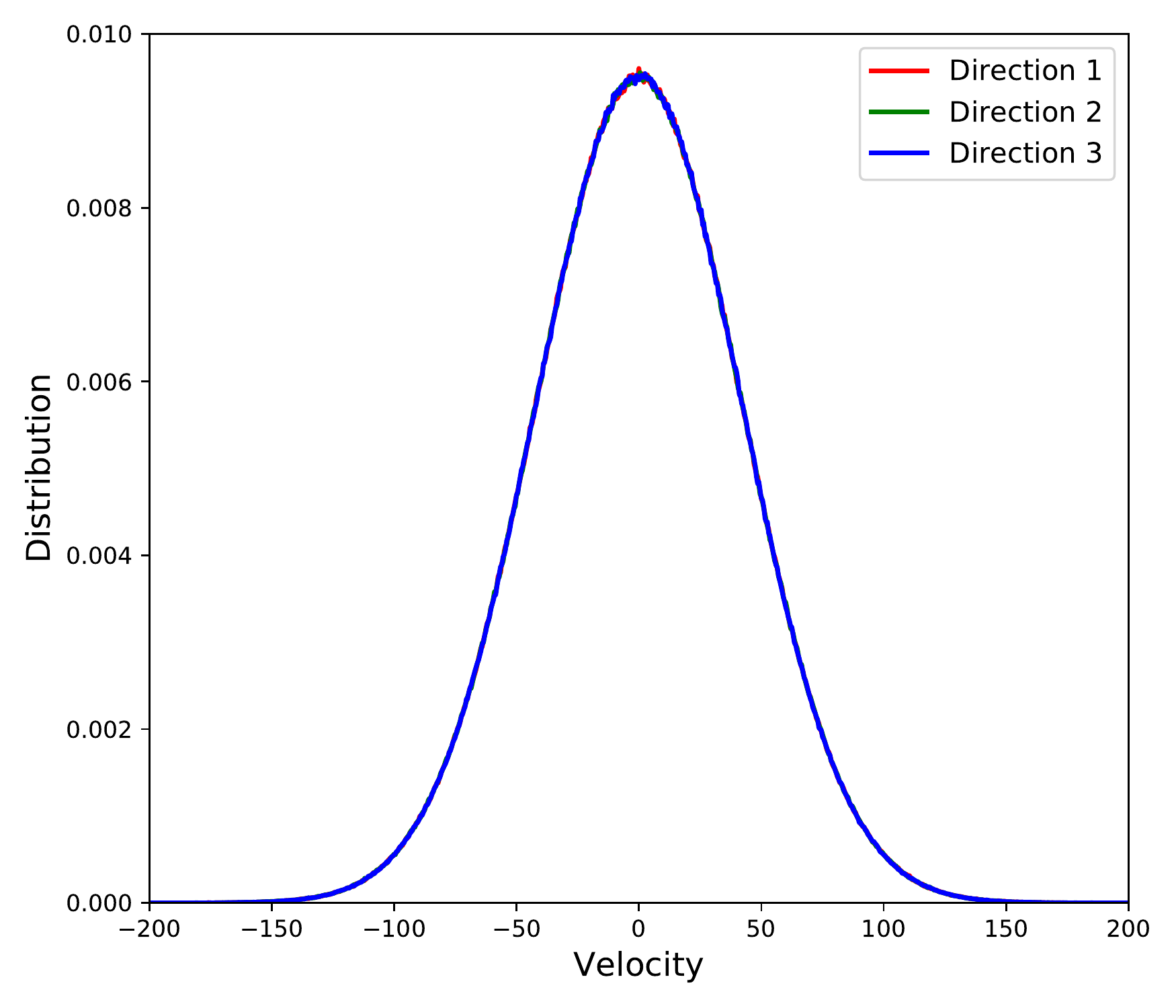}\hfill
\includegraphics[scale=0.45]{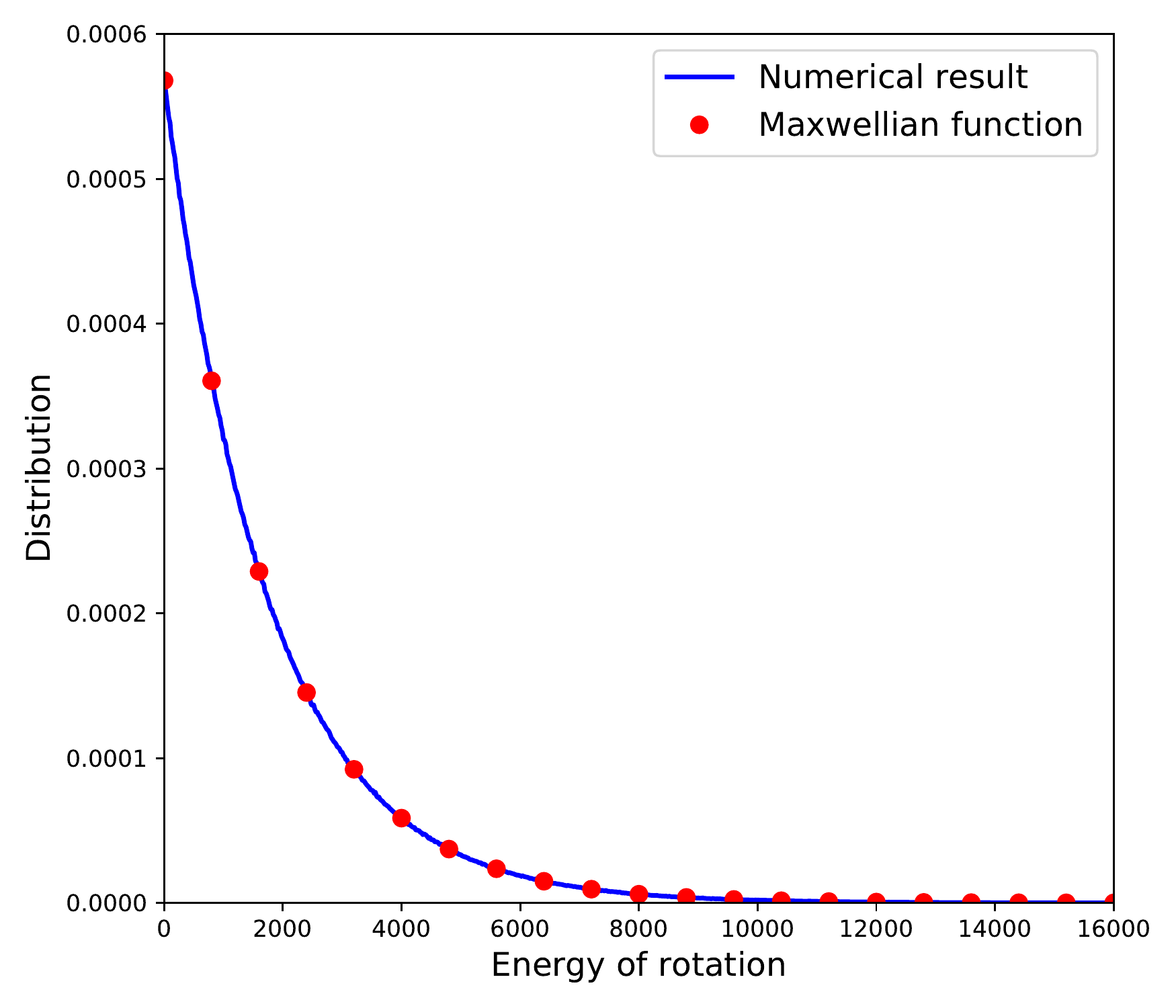}\\
\includegraphics[scale=0.45]{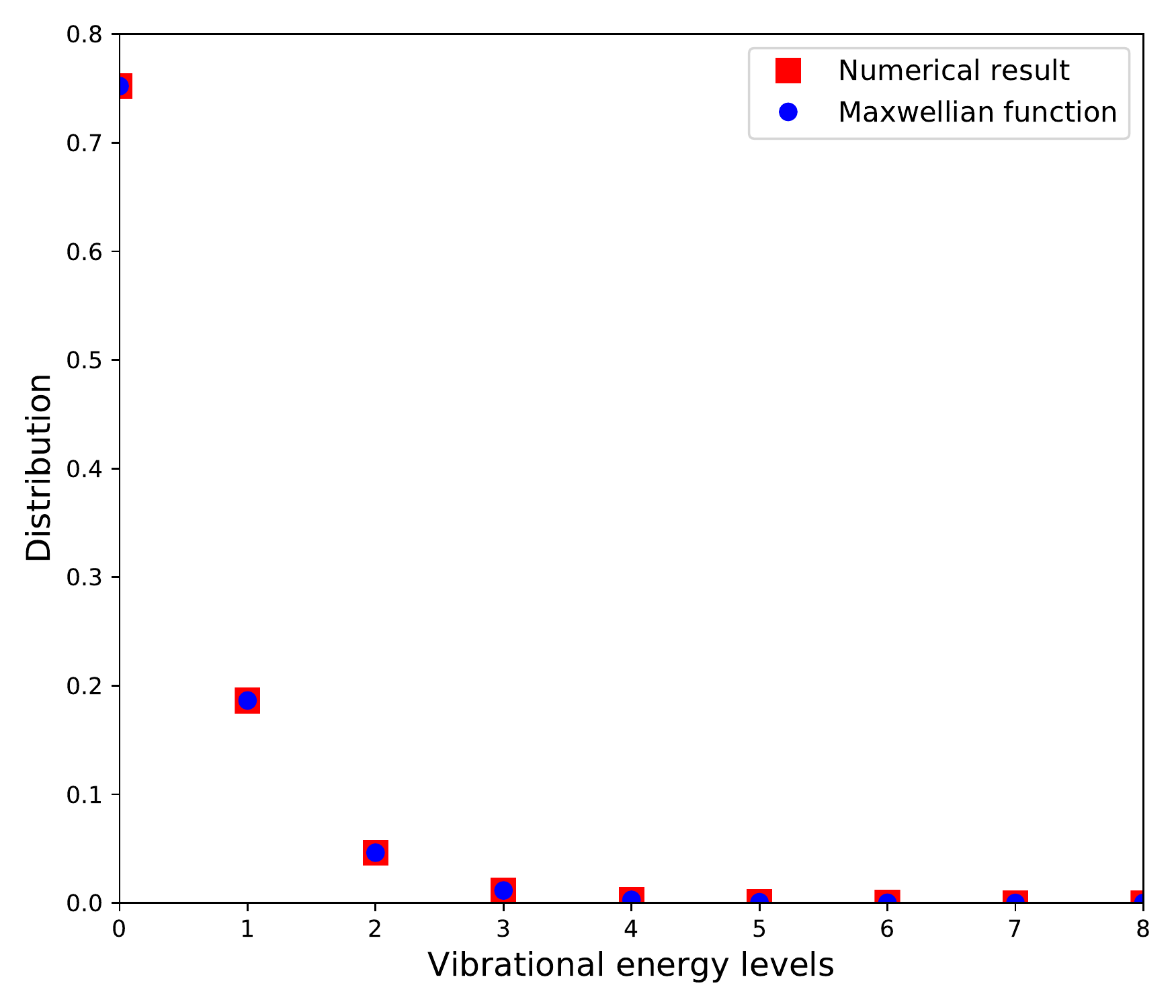}
\caption{\label{fig: equil}Top left: Distribution function of
  velocities at equilibrium:   $x$ direction (blue),   $y$ direction
  (red) and   $z$ direction (green). Top right: Distribution
  of the energy of rotation:  numerical result (blue) and equilibrium theoretical
  distribution (red). Bottom: discrete distribution of the
  vibrational energy:  numerical result (blue) and theoretical
  result (red).}
\end{figure}

\begin{figure}[htbp]
\centering 
\includegraphics[scale=0.58]{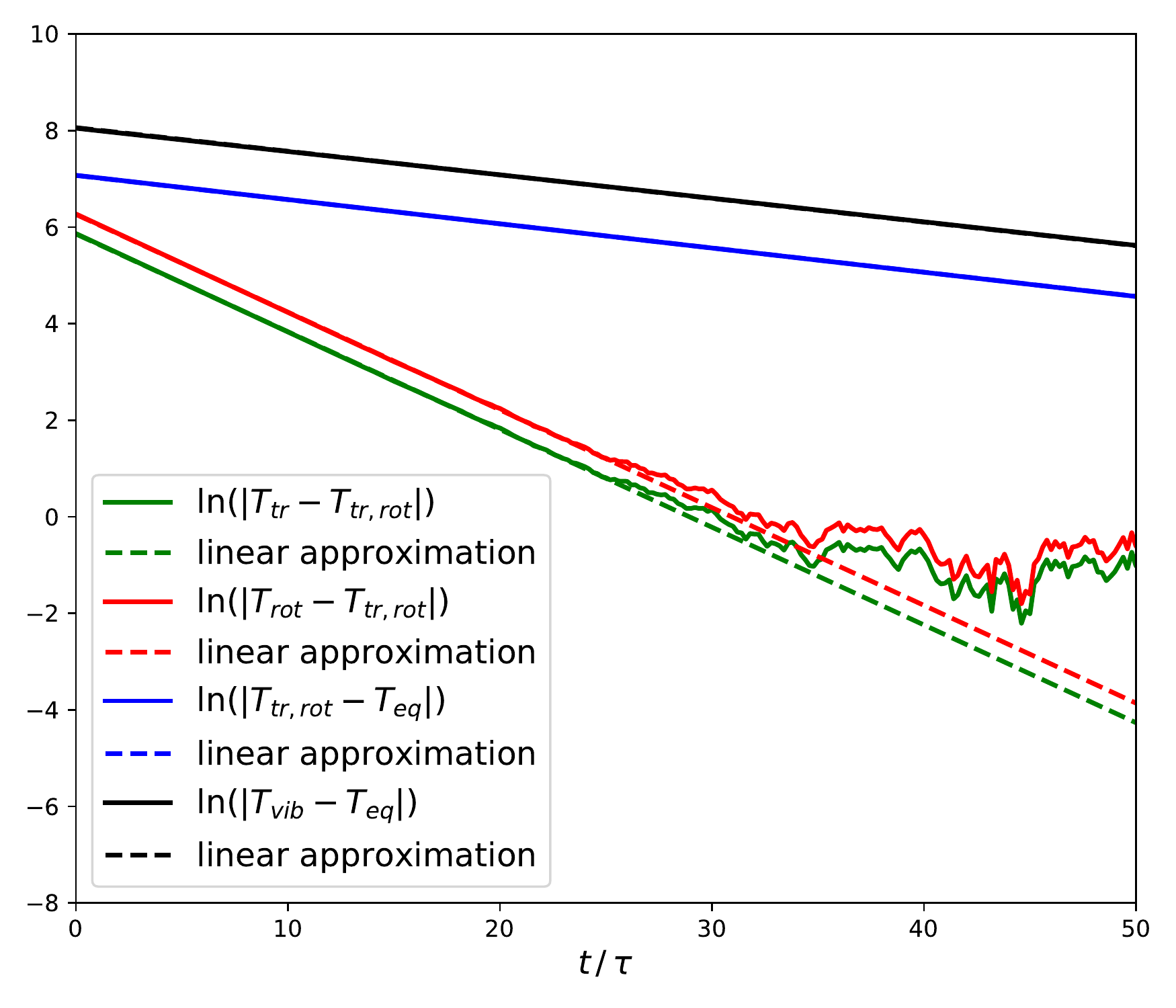}
\caption{\label{fig:diff_temp}Relaxation of  temperatures differences: $(T_{tr}-T_{tr,rot})$ (green),  $(T_{rot}-T_{tr,rot})$
  (red), $(T_{tr,rot}-T_{eq})$  (blue),  $(T_{vib}-T_{eq})$ (black).}
\end{figure}

\begin{figure}[htbp]
\centering 
\includegraphics[scale=0.58]{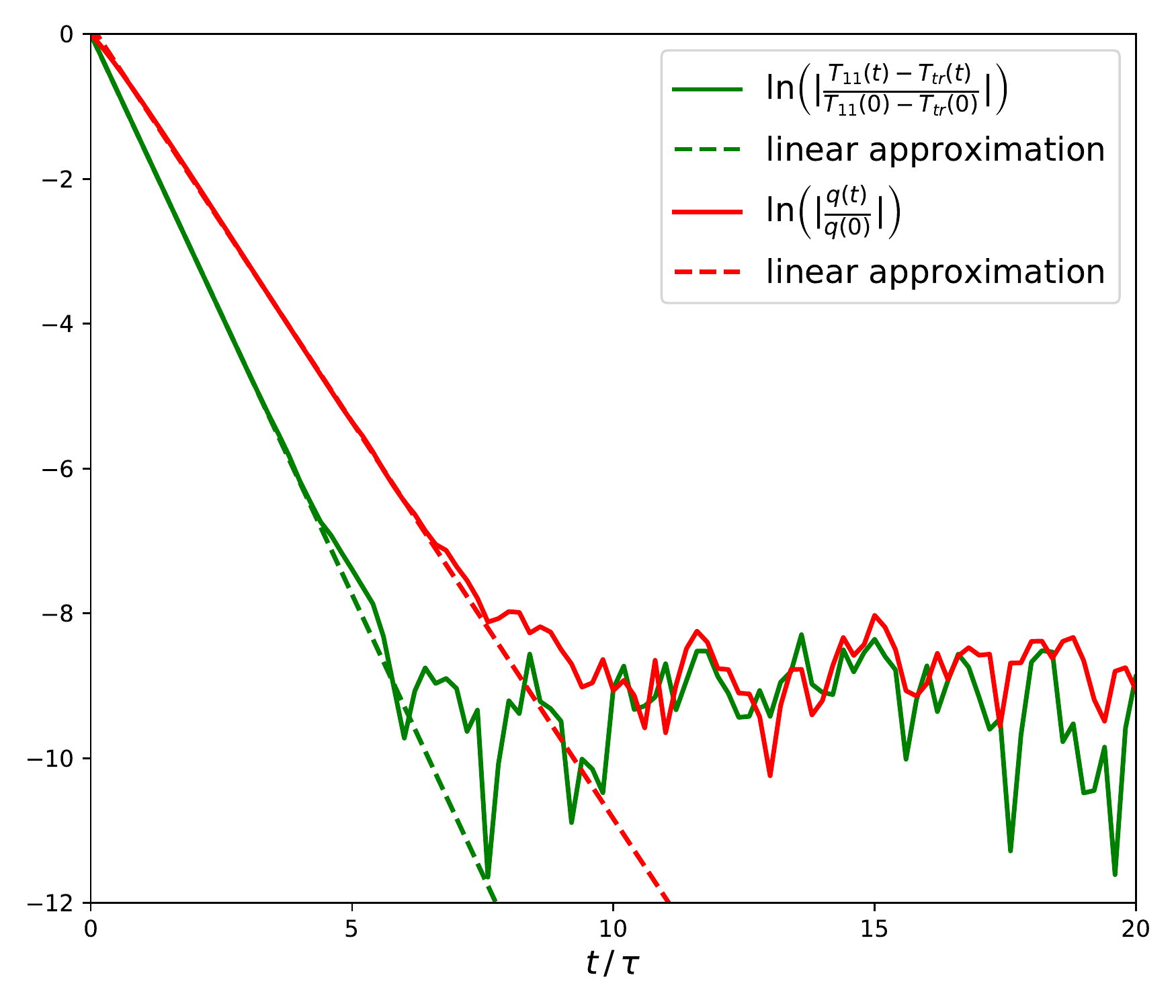}
\caption{\label{fig: Pr}Relaxation of the difference of temperatures $(T_{11}-T_{tr})$ (green) and  first component of the heat flux $q_1$ (red).}
\end{figure}

\end{document}